\documentclass[journal]{IEEEtran}

\usepackage{graphicx}
\usepackage{amsmath,amssymb}
\usepackage{amsthm}
\usepackage{url}
\usepackage[style=ieee, backend=biber]{biblatex}
\newtheorem{lemma}{Lemma}
\usepackage[acronym, nopostdot, nomain]{glossaries} 
\usepackage{microtype} 
\newacronym{EC}{EC}{Edge Computing}
\newacronym{XEC}{XEC}{Extreme Edge Computing}
\newacronym{IoT}{IoT}{Internet of Things}
\newacronym{CC}{CC}{Cloud Computing}
\newacronym{QoS}{QoS}{Quality of Service}
\newacronym{URLLC}{URLLC}{Ultra-Reliable low Latency Communication}
\newacronym{AR}{AR}{Augmented Reality}
\newacronym{DI}{DI}{Distributed Inference}
\newacronym{VR}{VR}{Virtual Reality}
\newacronym{TI}{TI}{Tactile Internet}
\newacronym{SDN}{SDN}{Software-Defined Networking}
\newacronym{MEC}{MEC}{Mobile Edge Computing}
\newacronym{XED}{XED}{Extreme Edge Device}
\newacronym{AI}{AI}{Artificial Intelligence}
\newacronym{DL}{DL}{Deep Learning}
\newacronym{ML}{ML}{Machine Learning}
\newacronym{CV}{CV}{Computer Vision}
\newacronym{NLP}{NLP}{Natural Language Processing}
\newacronym{KKT}{KKT}{Karush–Kuhn–Tucker}
\newacronym{ILP}{ILP}{Integer Linear Program}
\newacronym{MILP}{MILP}{Mixed-Integer Linear Program}
\newacronym{LLM}{LLM}{Large Language Model}
\newacronym{FL}{FL}{Federated Learning}
\newacronym{GPD}{GPD}{Generalized Pareto Distribution}
\newacronym{MI}{MI}{Minimal Information} 
\newacronym{MLE}{MLE}{Maximum Likelihood Estimation} 
\newacronym{PDF}{PDF}{Probability Density Function} 
\newacronym{CDF}{CDF}{Cumulative Distribution Function} 
\newacronym{RV}{RV}{Random Variable} 
\newacronym{OU}{OU}{Ornstein-Uhlenbeck} 

\title{How Reliable is Your Service at the Extreme Edge? Analytical Modeling of Computational Reliability}

\author{MHD Saria Allahham,~\IEEEmembership{Student Member,~IEEE,} and Hossam S. Hassanein,~\IEEEmembership{Senior Member,~IEEE}
\thanks{M. S. Allahham and H. S. Hassanein are with the School of Computing, Queen's University, Kingston, ON, Canada (e-mail: 20msa7@queensu.ca; hossam@cs.queensu.ca).}}

\begin{document}
	\pagestyle{plain}
	\maketitle
	
	\begin{abstract}
		Extreme Edge Computing (XEC) distributes streaming workloads across consumer-owned devices, exploiting their proximity to users and ubiquitous availability. Many such workloads are AI-driven, requiring continuous neural network inference for tasks like object detection and video analytics. Distributed Inference (DI), which partitions model execution across multiple edge devices, enables these streaming services to meet strict throughput and latency requirements. Yet consumer devices exhibit volatile computational availability due to competing applications and unpredictable usage patterns. This volatility poses a fundamental challenge: how can we quantify the probability that a device, or ensemble of devices, will maintain the processing rate required by a streaming service? This paper presents an analytical framework for computational reliability in XEC, defined as the probability that instantaneous capacity meets demand at a specified Quality of Service (QoS) threshold. We derive closed-form reliability expressions under two information regimes: Minimal Information (MI), requiring only declared operational bounds, and historical data, which refines estimates via Maximum Likelihood Estimation from past observations. The framework extends to multi-device deployments, providing reliability expressions for series, parallel, and partitioned workload configurations. We derive optimal workload allocation rules and analytical bounds for device selection, equipping orchestrators with tractable tools to evaluate deployment feasibility and configure distributed streaming systems. We validate the framework using real-time object detection with YOLO11m model as a representative DI streaming workload; experiments on emulated XED environments demonstrate close agreement between analytical predictions, Monte Carlo sampling, and empirical measurements across diverse capacity and demand configurations.
	\end{abstract}

	\section{Introduction}
	\label{sec:intro}
	
	The growth in \gls{IoT} devices and demand for real-time data processing is reshaping distributed computing. Traditional \gls{CC} introduces latency that is prohibitive for time-sensitive applications due to centralized processing \cite{1_cao2020overview}. Content Delivery Networks (CDNs) cache static content closer to users but are ineffective for dynamic, computation-heavy workloads that require continuous data generation and processing rather than retrieval. \gls{EC} addresses this by shifting computation toward the network periphery, reducing round-trip latency by processing data near its source \cite{2_mao2017survey, mach2017mec_survey}. However, the server-centric edge model struggles to meet the demands of AI-powered applications such as real-time \gls{AR} \cite{chen2019ar_offload, chen2021ar_drl}, cloud gaming \cite{cai2023cloud_gaming, carrascosa2024cloud_gaming_latency}, and video analytics \cite{lin2020edge_video_survey, yi2017lavea}. These applications perform deep neural network inference for object detection, semantic segmentation, pose estimation, and scene understanding, imposing \gls{URLLC}-grade demands that require millisecond-scale latencies and sustained computational throughput, which can overwhelm individual edge servers when serving multiple concurrent users. \gls{DI}, where deep neural network execution is distributed across nearby devices, whether by partitioning input data or the model itself, addresses this capacity constraint. For video streaming services, DI enables parallel frame processing across local device pools, sustaining throughput that individual devices cannot achieve alone \cite{23_7946410, ibrahim_extra}.
	
	These limitations have driven the evolution toward \gls{XEC} \cite{4_8607067, sherif1, 10_tourani2020democratizing}, which decentralizes computation by utilizing consumer-owned \glspl{XED} such as smartphones, tablets, wearables, and \gls{IoT} sensors \cite{5_covi2021adaptive}. This paradigm shares roots with volunteer computing \cite{mengistu2019volunteer, anderson2004boinc}, which demonstrated the viability of harnessing distributed, non-dedicated resources for large-scale computation. Unlike volunteer computing that targets batch scientific workloads, \gls{XEC} targets latency-sensitive streaming services where devices contribute resources \textit{opportunistically} to nearby requesters. Related concepts include femtoclouds \cite{mtibaa2015femtocloud} and device-to-device (D2D) assisted edge computing \cite{he2022d2d_social, feng2022d2d_iiot}, which leverage proximity for mobile computing. More broadly, Computing Power Networks (CPNs) \cite{sun2024cpn_survey} provide a unified framework for orchestrating heterogeneous computing resources across cloud, edge, and end-device tiers. In CPNs, \gls{XEC} constitutes the outermost layer of this architecture, where consumer devices contribute computational capacity to the network. 
	

	However, \gls{XEC} introduces challenges when it comes to providing a reliable service. The characteristics of \glspl{XED}, including mobility, fluctuating computational capacities, unpredictable usage patterns, and limited observability, create a volatile computational environment \cite{li2015service, 6_eec_re_10019108}. Unlike enterprise-owned edge servers with predictable resource availability, consumer devices experience capacity fluctuations driven by concurrent applications, background processes, battery management, and thermal throttling. This raises the question: \textit{how can an orchestrator quantify the probability that an XED will sustain the computational performance required by a streaming service?}
	
	Existing research does not provide a formal analytical framework to quantify this computational reliability. Prior work on reliability in edge systems assumes: (a) stable, enterprise-owned resources; (b) discrete task models rather than continuous streaming demands; or (c) reliability definitions based on node/link failure rather than sustained computational performance against a \gls{QoS} threshold. In \gls{XEC}-enabled streaming, the concern shifts from binary node outages to \textit{computational reliability}: the probability that an \gls{XED} sustains processing of a continuous data stream according to specified \gls{QoS} requirements.
	
	This paper proposes an analytical framework for quantifying computational reliability of streaming services in \gls{XEC}. We model reliability from a \textbf{computational perspective}: the probability that a worker's available computational capacity meets or exceeds the service's instantaneous demand, satisfying the specified \gls{QoS} requirement. The framework addresses scenarios with varying information availability, from minimal operational bounds to refined historical observations, and scales from individual device analysis to multi-device configurations. The contributions are:
	\begin{itemize}
		\item We derive closed-form expressions for single-device computational reliability under two information regimes, enabling tractable reliability assessment: the \gls{MI} model requires only declared capacity and demand bounds, while the historical model refines estimates via \gls{MLE} with provable convergence to true reliability.

		\item We extend the analysis to multi-device systems and derive optimal workload allocation rules that maximize system reliability for series, parallel, and work-partitioned configurations across heterogeneous device pools.

		\item We establish analytical bounds for device selection: the maximum number of devices supportable in series configurations and the minimum required for parallel redundancy, given a target reliability threshold for a streaming service.

		\item We validate the framework experimentally using YOLO11m inference as a representative \gls{DI} workload, demonstrating close agreement between analytical predictions and empirical measurements across diverse operating conditions.
	\end{itemize}
	
	The rest of the paper is organized as follows: Section~\ref{sec:related_work} reviews related work. Section~\ref{sec:sys_model} presents the system model. Section~\ref{sec:modeling_streaming} develops the single-device reliability framework. Section~\ref{sec:system_reliability} extends the analysis to multi-device systems. Section~\ref{sec:sim_exp} validates the models through simulation, and Section~\ref{sec:conc} concludes.
	
	\section{Related Work}
	\label{sec:related_work}
	
	\begin{figure*}
	    \centering
	    \includegraphics[scale=0.65]{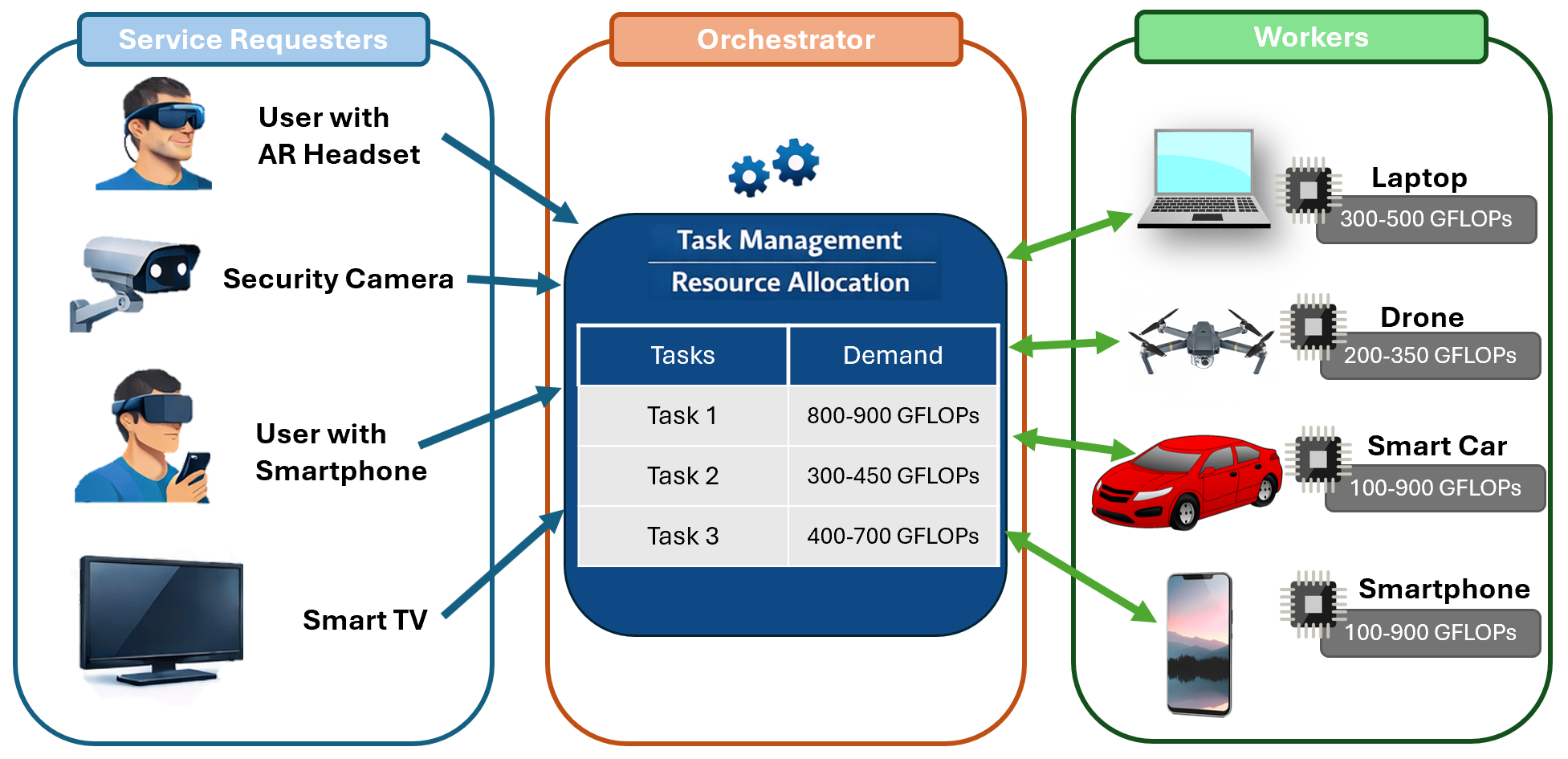}
	    \caption{The XEC system model}
	    \label{fig:sys_model}
	\end{figure*}
	The evolution from cloud to edge computing has generated research on task offloading and resource management. Surveys \cite{mach2017mec_survey, hazra2023resource_mec, wang2024task_offload_survey} cover MEC architectures, offloading strategies, and optimization objectives. In ultradense IoT networks, Guo et al. \cite{guo2018meco_ultradense} proposed game-theoretic offloading schemes for mobile-edge computation across heterogeneous base stations. These works primarily focus on minimizing latency or energy consumption under resource constraints, with reliability often treated as a secondary concern or assumed to be guaranteed by the underlying infrastructure.
	
	Video analytics at the edge has received attention due to its computational intensity and latency sensitivity. Yi et al. \cite{yi2017lavea} proposed LAVEA, a latency-aware video analytics platform that optimizes offloading decisions based on network conditions and server load. Surveys on edge video analytics \cite{lin2020edge_video_survey} cover applications, systems, and enabling techniques, highlighting the trade-offs between accuracy, latency, and resource consumption. Inference serving systems such as Jellyfish \cite{jellyfish2024inference} address end-to-end latency SLOs over dynamic edge networks, but focus on managed edge infrastructure rather than consumer-owned devices.
	
	For latency-sensitive applications like AR and cloud gaming, researchers have explored specialized offloading strategies. Chen et al. \cite{chen2019ar_offload, chen2021ar_drl} developed code-partitioning and deep reinforcement learning approaches for AR task offloading in MEC. Deep reinforcement learning has also been applied to optimize offloading policies in virtual edge computing systems without prior knowledge of network dynamics \cite{chen2019drl_offloading}. Cloud gaming research \cite{cai2023cloud_gaming, carrascosa2024cloud_gaming_latency} has addressed latency challenges through hybrid edge-cloud configurations and adaptive QoE models. However, these works assume dedicated edge servers with predictable resource availability, which does not hold in XEC environments.

	Reliability in distributed systems has traditionally been modeled through failure probabilities of individual components \cite{soi1981reliability_distributed, dai2014reliability_correlated}. Early work by Soi and Aggarwal \cite{soi1981reliability_distributed} defined distributed program reliability as the probability of successful execution requiring cooperation of multiple computers. Dai et al. \cite{dai2014reliability_correlated} extended this to heterogeneous systems with correlated failures, using integrated analytical and Monte Carlo approaches.
	
	Edge computing introduces distinct dependability challenges. Gill et al. \cite{gill2020dependability_edge} surveyed dependability challenges in edge computing, identifying resource constraints, device heterogeneity, and dynamic workloads as factors affecting reliability. Recent work on mission-critical applications \cite{chen2025mission_critical} reviewed resource management approaches that ensure reliability and application availability through fault-tolerant mechanisms.
	
	Runtime reliability monitoring has also been studied. Wang et al. \cite{wang2023bdetection} proposed B-Detection, which uses LSTM autoencoders to detect runtime reliability anomalies in MEC services. Deep learning has been applied to predict traffic dynamics in SDN-IoT networks, enabling adaptive resource allocation under varying load conditions \cite{tang2018sdn_iot_dl}. Chen et al. \cite{chen2023enhancing_resilience} addressed processing uncertainty in MEC through chance-constrained programming that offers probabilistic guarantees to task deadlines. These works provide insights into reliability detection and probabilistic guarantees but do not offer analytical frameworks for quantifying computational reliability of streaming services.
	
	Fault tolerance in edge computing commonly employs task replication and migration strategies. Nikolov et al. \cite{nikolov2022fault_pipeline} proposed automated fault tolerance pipelines using proactive horizontal scaling and node replication. Biswas et al. \cite{biswas2018reliable_iot} developed a reliable IoT-edge architecture that survives edge server failures through data replication and application redirection. Sun et al. \cite{sun2023fault_tolerance_survey} surveyed software fault tolerance in real-time systems, covering N-modular redundancy techniques where tasks are replicated and outputs are compared through voting mechanisms. These approaches enhance resilience but assume binary operational/failed states rather than continuous performance fluctuations.

	The concept of harnessing consumer devices for distributed computing has roots in volunteer computing. Mengistu and Che \cite{mengistu2019volunteer} surveyed volunteer computing systems, highlighting challenges in fault tolerance due to intermittent availability of resources when volunteers disconnect/connect at any time. BOINC \cite{anderson2004boinc} demonstrated large-scale public-resource computing but targeted batch scientific workloads rather than latency-sensitive streaming.
	
	Recent work has explored cooperative computing at the extreme edge. Zhang et al. \cite{union2023cooperative} proposed UNION, a fault-tolerant cooperative computing framework for opportunistic mobile edge clouds that estimates the probability of successful execution for specific opportunistic paths. This work is closest to ours in addressing reliability of opportunistic resources, but focuses on task-level success probability rather than sustained computational performance against QoS thresholds for streaming services.
	
	D2D-assisted mobile edge computing \cite{he2022d2d_social, feng2022d2d_iiot} leverages device proximity for task offloading. He et al. \cite{he2022d2d_social} incorporated social relationships to improve D2D link reliability, while Feng et al. \cite{feng2022d2d_iiot} studied task co-offloading in industrial IoT. These works optimize offloading decisions assuming known resource availability rather than modeling the stochastic computational capacity of consumer devices.

	Handling resource uncertainty in edge computing has been addressed through various stochastic optimization approaches. Liu et al. \cite{liu2020qos_statistical} argued that deterministic QoS guarantees are impractical due to high dynamics of wireless environments and instead focused on task offloading with statistical QoS guarantees where tasks complete before deadlines with specified probability. Apostolopoulos et al. \cite{apostolopoulos2023data_offload_uncertainty} addressed data offloading under resource uncertainty in UAV-assisted MEC, while Chen et al. \cite{chen2020energy_stochastic} modeled energy-aware application placement as a multi-stage stochastic program. Deep reinforcement learning has been applied to jointly optimize mode selection and resource management in fog radio access networks under dynamic edge cache states \cite{sun2019green_fog}. Franceschetti et al. \cite{pacca2023stochastic} proposed probabilistic admission control for managing edge offloading with stochastic workloads and deadlines.
	
	Task scheduling research has also addressed reliability concerns. Wang et al. \cite{wang2023reliability_aware} jointly studied reliability-aware data compression and task offloading for data-intensive applications in MEC, formulating optimization problems to minimize latency while satisfying reliability constraints. Chen et al. \cite{chen2023decentralized_scheduling} applied deep reinforcement learning for decentralized task scheduling in MEC. Joint caching and resource allocation in cooperative MEC systems, where multiple edge servers collaboratively serve user requests, has been optimized using hierarchical reinforcement learning \cite{zhang2024coop_mec}. Costa et al. \cite{costa2024cost_aware} addressed cost-aware service placement in the edge-cloud continuum.


	\section{System Model}
	\label{sec:sys_model}
	We consider an \gls{XEC} environment where consumer-owned devices (smartphones, tablets, wearables) process \gls{DI} workloads for streaming applications. Devices assume one of three roles: \textit{Service Requester} (originates processing demand), \textit{Worker} (provides computational resources), or \textit{Orchestrator} (coordinates placement and monitors \gls{QoS}). The orchestrator's challenge is selecting workers that reliably sustain the required performance despite resource volatility~\cite{6_eec_re_10019108}.

	Let $C_i(t)$ denote worker $i$'s available computational capacity (e.g., FLOPS) at time $t$, and $\Delta_i(t)$ the service's computational demand per unit time assigned to that worker. Both are stochastic processes: $C_i(t)$ fluctuates due to concurrent applications, background processes, and thermal throttling; $\Delta_i(t)$ varies with input complexity (e.g., scene density in object detection). This inherent variability precludes point-value capacity commitments (i.e., a worker cannot promise exactly how much capacity it will provide); workers instead declare operational bounds $[C_i^{\min}, C_i^{\max}]$, and services specify demand bounds $[\Delta_i^{\min}, \Delta_i^{\max}]$. We model $C_i(t)$ and $\Delta_i(t)$ as independent random variables at each time instant, which holds when device-local capacity fluctuations are uncorrelated with input-driven demand variations.

	The \gls{QoS} threshold $\Theta \geq 1$ specifies the minimum acceptable ratio $C_i(t)/\Delta_i(t)$. This ratio directly governs whether real-time deadlines are met. Consider video analytics at $f$ frames per second: each frame must complete within $1/f$ seconds. If one frame requires $\Delta_i(t)$ operations, sustaining the frame rate requires $C_i(t) \geq f \cdot \Delta_i(t)$, i.e., $C_i(t)/\Delta_i(t) \geq f$. More generally, $\Theta$ encodes processing headroom: $\Theta = 1$ means capacity exactly matches demand, while $\Theta > 1$ provides margin against variability. We define \textit{computational reliability} of worker $i$ as the probability that the \gls{QoS} condition is satisfied:
	\begin{equation}
		R_i(t) = P\left(\frac{C_i(t)}{\Delta_i(t)} \geq \Theta\right).
		\label{eq:reliability_def}
	\end{equation}
	This metric captures the continuous spectrum of performance degradation inherent to consumer devices, unlike binary failure models that assume nodes are either fully operational or failed.


	\section{Modeling of Streaming Service Reliability}
	\label{sec:modeling_streaming}

	This section derives closed-form expressions for $R_i(t)$ under two information regimes: Minimal Information (MI), using only the declared bounds, and historical data, incorporating past observations via MLE.
	
	\subsection{Modeling with Minimal Information}
	\label{subsubsec:streaming_MI_main}
	
	When the orchestrator encounters a new worker or service with no interaction history, reliability estimation must proceed with minimal information: only the declared bounds $[C_i^{\min}, C_i^{\max}]$ for allocable capacity and $[\Delta_i^{\min}, \Delta_i^{\max}]$ for expected demand, which workers must advertise as a prerequisite for participation in XEC. Without prior observations, the orchestrator has no basis to favor any value within these intervals.

	This absence of prior knowledge motivates uniform distributions as the natural modeling choice. The uniform distribution represents the constraint that values lie within bounds while making no additional assumptions. We model $C_i(t)$ and $\Delta_i(t)$ as independent, uniformly distributed random variables:
	$C_i(t) \sim \text{U}(C_i^{\min}, C_i^{\max})$ and $\Delta_i(t) \sim \text{U}(\Delta_i^{\min}, \Delta_i^{\max})$.
	Let $C_i^{\text{range}} = C_i^{\max} - C_i^{\min}$ and $\Delta_i^{\text{range}} = \Delta_i^{\max} - \Delta_i^{\min}$. The probability density functions (PDFs) are $f_{C_i}(c) = 1/C_i^{\text{range}}$ for $c \in [C_i^{\min}, C_i^{\max}]$, and $f_{\Delta_i}(\delta) = 1/\Delta_i^{\text{range}}$ for $\delta \in [\Delta_i^{\min}, \Delta_i^{\max}]$.

	The reliability $R^{\text{MI}}_i(t, \Theta) = P(C_i(t)/\Delta_i(t) \geq \Theta)$ is derived by integrating the joint PDF $f_{C_i,\Delta_i}(c,\delta) = f_{C_i}(c)f_{\Delta_i}(\delta)$ over the region defined by $C_i^{\min} \leq c \leq C_i^{\max}$, $\Delta_i^{\min} \leq \delta \leq \Delta_i^{\max}$, and $c \geq \Theta \delta$. The general integral form is:
	\begin{equation}
		R^{\text{MI}}_i(t, \Theta) = \int_{\Delta_i^{\min}}^{\Delta_i^{\max}} \int_{\max(C_i^{\min}, \Theta \delta)}^{C_i^{\max}} \frac{dc \, d\delta}{C_i^{\text{range}}\Delta_i^{\text{range}}}.
		\label{eq:R_MI_Theta_integral_form_main}
	\end{equation}
	This integral requires a careful case-by-case analysis based on the relative values of $\Delta_i^{\min}, \Delta_i^{\max}, C_i^{\min}/\Theta$, and $C_i^{\max}/\Theta$. A specific form of this integral, often used under simplifying assumptions, is presented below.
	
	\begin{figure*}[!t]
		\normalsize
		\begin{equation}
			\label{eq:MI_Theta_reliability_simplified_formatted}
			\begin{split}
				R^{\text{MI}}_i(t, \Theta) = \frac{1}{(C_i^{\max} - C_i^{\min}) (\Delta_i^{\max} - \Delta_i^{\min})} \Biggl[ & \left( C_i^{\max} \min\left(\Delta_i^{\max}, \frac{C_i^{\max}}{\Theta}\right) - \frac{\Theta}{2} \min\left(\Delta_i^{\max}, \frac{C_i^{\max}}{\Theta}\right)^2 \right) \\
				& - \left( C_i^{\max} \Delta_i^{\min} - \frac{\Theta (\Delta_i^{\min})^2}{2} \right) \Biggr]
			\end{split}
		\end{equation}
		\hrulefill
		\vspace*{4pt}
	\end{figure*}
	
	\begin{lemma}
		\label{lemma:streaming_MI_reliability_main}
		For an XED $i$ with computational capacity $C_i(t) \sim \text{U}(C_i^{\min}, C_i^{\max})$ and a streaming service with demand $\Delta_i(t) \sim \text{U}(\Delta_i^{\min}, \Delta_i^{\max})$, where $C_i(t)$ and $\Delta_i(t)$ are independent, if the condition $C_i^{\min} \leq \Theta \delta$ holds for the entire integration range of $\delta$ from $\Delta_i^{\min}$ to $\min(\Delta_i^{\max}, C_i^{\max}/\Theta)$, then the reliability $R^{\text{MI}}_i(t, \Theta)$ with QoS threshold $\Theta > 0$ is given by Eq.~\eqref{eq:MI_Theta_reliability_simplified_formatted}.
	\end{lemma}
	\begin{proof}
		See Appendix~\ref{app:proof_lemma1}.
	\end{proof}
	This lemma provides the orchestrator with a closed-form reliability estimate for worker $i$ when no historical data is available. Given only the declared capacity and demand bounds, the orchestrator can immediately compute the probability that the worker meets the QoS requirement, without requiring any prior interactions or observations.

	\subsection{Modeling with Historical Data}
	\label{subsubsec:streaming_h_main}

	When the orchestrator has accumulated observations from prior interactions with worker $i$ or service, these historical data allow more accurate reliability estimation. Let $\{c_{i,1}, \dots, c_{i,t}\}$ and $\{\delta_{i,1}, \dots, \delta_{i,t}\}$ denote observed capacity and demand samples collected up to time $t$. With historical samples available, the orchestrator can estimate distributional parameters that better characterize system behavior. However, any estimated distribution must respect the declared operational bounds $[C_i^{\min}, C_i^{\max}]$ and $[\Delta_i^{\min}, \Delta_i^{\max}]$, which remain physical constraints regardless of observed data. Without loss of generality, we estimate the mean and variance from the samples, leading to a truncated normal distribution that captures the observed behavior while enforcing the declared bounds. The framework generalizes: depending on the application or environment, any parameters can be estimated via \gls{MLE} from the data, and any bounded distribution can be selected.

	Applying \gls{MLE} to these observations yields time-varying distributional parameters $(\mu_{C_i}(t), \sigma_{C_i}(t))$ and $(\mu_{\Delta_i}(t), \sigma_{\Delta_i}(t))$, updated as new observations arrive. This online refinement progressively reduces estimation uncertainty compared to the bounds-only MI approach. We model capacity $C_i(t)$ and demand $\Delta_i(t)$ as independent, truncated normal random variables bounded by the MI-specified limits:
	\begin{align}
		C_i(t) &\sim \text{TNORM}(\mu_{C_i}(t), \sigma_{C_i}^2(t), C_i^{\min}, C_i^{\max}), \nonumber \\
		\Delta_i(t) &\sim \text{TNORM}(\mu_{\Delta_i}(t), \sigma_{\Delta_i}^2(t), \Delta_i^{\min}, \Delta_i^{\max}). \nonumber
	\end{align}
	Let $\phi(\cdot)$ and $\Phi(\cdot)$ denote the PDF and CDF of the standard normal distribution, respectively.
	The PDF of $C_i(t)$ is $f_{C_i}(c; t) = \frac{1}{\sigma_{C_i}(t) Z_{C_i}(t)} \phi\left(\frac{c - \mu_{C_i}(t)}{\sigma_{C_i}(t)}\right)$ for $c \in [C_i^{\min}, C_i^{\max}]$, where $Z_{C_i}(t) = \Phi\left(\frac{C_i^{\max} - \mu_{C_i}(t)}{\sigma_{C_i}(t)}\right) - \Phi\left(\frac{C_i^{\min} - \mu_{C_i}(t)}{\sigma_{C_i}(t)}\right)$ is the normalization constant. A similar expression holds for $f_{\Delta_i}(\delta; t)$.

	Since the MLE parameters evolve with accumulated observations, the reliability estimate is time-dependent:

	\begin{equation}
		\begin{split}
			R^{\text{H}}_i(t, \Theta) &= P(C_i(t)/\Delta_i(t) \geq \Theta)= \\
			& \int_{\Delta_i^{\min}}^{\Delta_i^{\max}} \int_{\max(C_i^{\min}, \Theta \delta)}^{C_i^{\max}} f_{C_i}(c; t)f_{\Delta_i}(\delta; t) \,dc \,d\delta.
		\end{split}
		\label{eq:R_H_Theta_integral_form_general_main}
	\end{equation}

	As more observations are collected, the MLE estimates converge to the true distributional parameters, and $R^{\text{H}}_i(t, \Theta)$ converges to the true reliability. Under the simplifying assumption that $\Theta \delta \ge C_i^{\min}$ for the relevant range of $\delta$, the inner integral limit $\max(C_i^{\min}, \Theta \delta)$ becomes $\Theta \delta$.

	\begin{figure*}[!t]
		\normalsize
		\begin{equation}
			R^{\text{H}}_i(t, \Theta) = \frac{1}{Z_{C_i}(t) Z_{\Delta_i}(t)} \left[ \Phi\left(\frac{C_i^{\max} - \mu_{C_i}(t)}{\sigma_{C_i}(t)}\right) Z_{\Delta_i}(t) - \frac{1}{\sigma_{\Delta_i}(t)} \int_{\Delta_i^{\min}}^{\Delta_i^{\max}} \phi\left(\frac{\delta - \mu_{\Delta_i}(t)}{\sigma_{\Delta_i}(t)}\right) \Phi\left(\frac{\Theta \delta - \mu_{C_i}(t)}{\sigma_{C_i}(t)}\right) \,d\delta \right]
			\label{eq:H_Theta_reliability_presented_earlier_lemma_formatted}
		\end{equation}
		\hrulefill
		\vspace*{4pt}
	\end{figure*}

	\begin{lemma}
		\label{lemma:streaming_si_reliability_main}
		For worker $i$ with independent truncated normal capacity and demand distributions (as defined above), with MLE parameters updated from observations up to time~$t$, if $\Theta \delta \ge C_i^{\min}$ for all $\delta \in [\Delta_i^{\min}, \Delta_i^{\max}]$, then $R^{\text{H}}_i(t, \Theta)$ is given by Eq.~\eqref{eq:H_Theta_reliability_presented_earlier_lemma_formatted}, where $Z_{C_i}(t) = \Phi(\bar{C}_i^{\max}(t)) - \Phi(\bar{C}_i^{\min}(t))$, $Z_{\Delta_i}(t) = \Phi(\bar{\Delta}_i^{\max}(t)) - \Phi(\bar{\Delta}_i^{\min}(t))$, and $\bar{C}_i^{\max}(t) = (C_i^{\max}-\mu_{C_i}(t))/\sigma_{C_i}(t)$, etc.
	\end{lemma}
	\begin{proof}
		See Appendix~\ref{app:proof_lemma2}.
	\end{proof}
	This lemma provides a refined reliability estimate for worker $i$ that improves as the orchestrator accumulates observations. Unlike the MI approach that treats all values within bounds as equally likely, this formulation incorporates learned distributional parameters, yielding reliability estimates that converge toward the true value over time. The integral term requires numerical evaluation but can be computed efficiently for real-time decisions.

	\section{System-Level Reliability in XEC}
	\label{sec:system_reliability}

	The preceding analysis quantifies the computational reliability of a single XED. However, many AI inference workloads exceed individual device capacity, requiring multi-device deployments that distribute workloads across a pool of XEDs. This section extends the reliability framework to systems of $N$ XEDs, addressing three interconnected decisions: Is a multi-device configuration feasible for a target reliability? How many and which devices should be selected? How should workload be partitioned to maximize system reliability? We answer these questions through closed-form expressions derived herein.

	\subsection{Independence in XEC Systems}
	\label{subsec:independence}

	Let each device $i \in \{1, \ldots, N\}$ have reliability $R_i(t, \Theta)$, computed via the MI framework or historical data approach of Section~\ref{sec:modeling_streaming}. A fundamental modeling assumption is statistical independence of device successes. In XEC, this assumption reflects the nature of consumer device ownership: each XED belongs to a different user whose computational load depends on personal usage patterns, active applications, background processes, and battery management policies. Unlike enterprise data centers where correlated failures arise from shared infrastructure, XEDs experience capacity fluctuations driven by independent human behaviors. This independence yields tractable system reliability expressions and reflects the decentralized character of XEC resource pools.
	
	\subsection{Series Configuration}
	\label{subsec:series}
	
	In a series configuration, the streaming service is decomposed such that each device handles a distinct portion of the workload, and all portions must complete for the service to succeed. This architecture arises in spatial partitioning where a single frame is divided across devices, each processing a subregion. The service fails if any device fails to meet its QoS constraint, yielding system reliability as the product of individual reliabilities:
	\begin{equation}
		R_{\text{series}}(t) = \prod_{i=1}^{N} R_i(t, \Theta_i).
		\label{eq:series_reliability}
	\end{equation}
	This multiplicative structure imposes a strict penalty: for $N$ identical devices each with reliability $R(t, \Theta)$, system reliability decays exponentially as $R(t, \Theta)^N$. Consequently, series configurations in XEC demand either highly reliable devices or minimal partitioning.
	
	\subsection{Parallel Configuration}
	\label{subsec:parallel}
	
	In a parallel configuration, multiple XEDs redundantly process the same workload, and the service succeeds if at least one device meets the QoS constraint. For video streaming applications with strict frame rate requirements, replicating each frame across multiple devices ensures that at least one device completes processing within the deadline, maintaining the target frame throughput even when individual devices intermittently fail to meet QoS. The system reliability becomes:
	\begin{equation}
		R_{\text{parallel}}(t) = 1 - \prod_{i=1}^{N} (1 - R_i(t, \Theta)).
		\label{eq:parallel_reliability}
	\end{equation}
	Here, reliability improves with each additional device. For $N$ identical devices each with reliability $R(t, \Theta)$, the system achieves $1-(1-R(t, \Theta))^N$, which approaches unity as $N$ grows. This configuration trades computational redundancy for resilience, using multiple nearby consumer devices to achieve reliability through replication rather than individual device guarantees.
	
	\subsection{Work Partitioning}
	\label{subsec:work_partitioning}
	
	In XEC, the orchestrator can partition the total throughput requirement across devices (i.e., partitioning the input or the model reduces the computational load per device, effectively splitting the throughput target), where each device contributes a fraction of the target rather than meeting the full threshold independently. Let $\alpha_i > 0$ denote the fraction of the total requirement assigned to device $i$, with $\sum_{i=1}^{N} \alpha_i = 1$. Each device faces a reduced effective threshold $\alpha_i \Theta$, and the system succeeds when all devices meet their respective reduced thresholds:
	\begin{equation}
		R_{\text{dist}}(t) = \prod_{i=1}^{N} R_i(t, \alpha_i \Theta).
		\label{eq:distributed_reliability}
	\end{equation}
	This formulation captures scenarios where an orchestrator splits a video stream across devices by frame ranges, or distributes inference batches proportionally to device capabilities. Since reliability $R_i(t, \cdot)$ decreases monotonically with threshold, smaller $\alpha_i$ values increase per-device reliability. The orchestrator must balance reduced per-device load against the multiplicative penalty of additional failure points.

	\begin{lemma}[Optimal Work Partitioning]
		\label{lemma:optimal_partitioning}
		Consider $N$ devices with differentiable reliability functions $R_i(t, \cdot): \mathbb{R}_{>0} \to (0,1]$, each monotonically decreasing in the threshold argument. The allocation $\boldsymbol{\alpha}^* = (\alpha_1^*, \ldots, \alpha_N^*)$ maximizing $R_{\text{dist}}(t) = \prod_{i=1}^{N} R_i(t, \alpha_i \Theta)$ subject to $\sum_{i=1}^{N} \alpha_i = 1$ satisfies the equal marginal log-reliability condition:
		\begin{equation}
			\frac{\partial R_i(t, \alpha_i^* \Theta)/\partial \Theta}{R_i(t, \alpha_i^* \Theta)} = \frac{\partial R_j(t, \alpha_j^* \Theta)/\partial \Theta}{R_j(t, \alpha_j^* \Theta)} \quad \forall\, i, j \in \{1, \ldots, N\}.
			\label{eq:optimal_partitioning_condition}
		\end{equation}
	\end{lemma}
	\begin{proof}
		See Appendix~\ref{app:proof_lemma3}.
	\end{proof}
	This lemma provides the orchestrator with a principled rule for distributing workload across heterogeneous devices. Rather than using heuristics or equal splitting, the orchestrator should allocate work such that all devices experience the same marginal reliability degradation per unit of additional load. In practice, this means assigning more work to devices whose reliability degrades more gracefully with increased demand.

	\textit{Homogeneous devices:} When $R_i(t, \cdot) = R_j(t, \cdot)$ for all $i, j$, symmetry implies $\alpha_i^* = 1/N$. Equal partitioning is optimal when devices have identical reliability characteristics.

	\textit{Heterogeneous devices:} Devices whose reliability degrades more slowly with load (smaller $|\partial R_i(t, \theta)/\partial \theta / R_i(t, \theta)|$) receive larger workload shares $\alpha_i$, concentrating demand on the stronger devices.

	\subsection{Device Selection for Target Reliability}
	\label{subsec:device_selection}
	
	An orchestrator managing an XEC resource pool faces a selection problem: given $M$ candidate devices with reliabilities $R_1(t, \Theta), \ldots, R_M(t, \Theta)$ and a target system reliability $\varepsilon$, how many and which devices should be recruited? The answer depends on the configuration.

	\begin{lemma}[Series System Feasibility]
		\label{lemma:series_feasibility}
		Given $M$ candidate devices with reliabilities $R_1(t, \Theta), \ldots, R_M(t, \Theta)$, sorted in descending order $R_{(1)}(t, \Theta) \geq \cdots \geq R_{(M)}(t, \Theta)$, and target system reliability $\varepsilon \in (0,1)$. For a series system requiring exactly $N$ devices, selecting the $N$ most reliable devices is feasible if and only if:
		\begin{equation}
			\prod_{i=1}^{N} R_{(i)}(t, \Theta) \geq \varepsilon.
			\label{eq:series_feasibility}
		\end{equation}
		When the number of devices is flexible, the maximum $N^*_{\text{series}}$ satisfying target $\varepsilon$ is:
		\begin{equation}
			N^*_{\text{series}} = \max\left\{N \in \{1, \ldots, M\} : \prod_{i=1}^{N} R_{(i)}(t, \Theta) \geq \varepsilon\right\}.
			\label{eq:series_max_general}
		\end{equation}
	\end{lemma}
	\begin{proof}
		See Appendix~\ref{app:proof_lemma4}.
	\end{proof}
	This lemma answers a key feasibility question: can a series configuration of $N$ devices meet the target reliability? The orchestrator sorts candidate devices by reliability and checks whether the product of the top $N$ reliabilities exceeds $\varepsilon$. If not, the configuration is infeasible with available devices.

	\textit{Special case (uniform reliability):} When all devices have identical reliability $R_{(i)}(t, \Theta) = R(t, \Theta)$, the constraint becomes $R(t, \Theta)^N \geq \varepsilon$, yielding the closed-form bound $N_{\max}^{(\text{series})} = \lfloor \ln \varepsilon / \ln R(t, \Theta) \rfloor$.

	\begin{lemma}[Parallel System Device Selection]
		\label{lemma:parallel_selection}
		Given $M$ candidate devices with reliabilities $R_1(t, \Theta), \ldots, R_M(t, \Theta)$, sorted in descending order $R_{(1)}(t, \Theta) \geq \cdots \geq R_{(M)}(t, \Theta)$, and target system reliability $\varepsilon \in (0,1)$. Define the cumulative failure probability $P_N(t) = \prod_{i=1}^{N}(1 - R_{(i)}(t, \Theta))$. The minimum number of devices achieving $R_{\text{parallel}}(t) \geq \varepsilon$ is:
		\begin{equation}
			N^*_{\text{parallel}} = \min\left\{N \in \{1, \ldots, M\} : P_N(t) \leq 1 - \varepsilon\right\}.
			\label{eq:parallel_exact}
		\end{equation}
	\end{lemma}
	\begin{proof}
		See Appendix~\ref{app:proof_lemma5}.
	\end{proof}
	This lemma tells the orchestrator how many redundant devices are needed to achieve a target reliability through parallel redundancy. Starting with the most reliable devices, the orchestrator adds devices until the cumulative failure probability drops below the acceptable threshold.

	\textit{Special case (uniform reliability):} When $R_{(i)}(t, \Theta) = R(t, \Theta)$ for all $i$, we have $P_N(t) = (1-R(t, \Theta))^N$, and the constraint $(1-R(t, \Theta))^N \leq 1-\varepsilon$ yields $N^*_{\text{parallel}} = \lceil \ln(1-\varepsilon) / \ln(1-R(t, \Theta)) \rceil$.

	\textit{Worst-case bound:} When only a lower bound $R_{\min}(t, \Theta) = \min_i R_i(t, \Theta)$ is available (e.g., from~MI), we substitute $R_{(i)}(t, \Theta) \geq R_{\min}(t, \Theta)$ to obtain $P_N(t) \leq (1-R_{\min}(t, \Theta))^N$. A sufficient condition for $P_N(t) \leq 1-\varepsilon$ is $(1-R_{\min}(t, \Theta))^N \leq 1-\varepsilon$, giving the conservative upper bound:
	\begin{equation}
		N^*_{\text{parallel}} \leq N_{\text{upper}}^{(\text{parallel})} = \left\lceil \frac{\ln(1 - \varepsilon)}{\ln(1 - R_{\min}(t, \Theta))} \right\rceil.
		\label{eq:parallel_upper_bound}
	\end{equation}
	This supports resource provisioning under uncertainty without detailed per-device reliability estimates.
	
	\section{Case Study: Real-Time Object Detection}
	\label{sec:case_study}

	The reliability framework developed in Sections~\ref{sec:modeling_streaming}--\ref{sec:system_reliability} applies to any streaming service where computational demand and capacity can be characterized by bounded distributions. To illustrate and validate the framework, we examine real-time object detection using YOLO (You Only Look Once) as a representative \gls{DI} workload. YOLO is a family of single-stage detectors that perform localization and classification in a single forward pass \cite{redmon2016yolo}. Let $f_\phi: \mathbb{R}^{H \times W \times 3} \to \mathcal{B}$ denote the detection function parameterized by network weights $\phi$, mapping an RGB image $\mathcal{I}$ of spatial dimensions $H \times W$ to a set of bounding boxes $\mathcal{B}$. YOLO11, the latest iteration, achieves state-of-the-art accuracy while enabling real-time inference \cite{yolo11_ultralytics}.

	The computational cost of evaluating $f_\phi(\mathcal{I})$ is dominated by convolutional operations. For an input of size $H \times W$, each convolutional layer with kernel size $k \times k$ and $C_{\text{out}}$ output channels requires $\mathcal{O}(k^2 \cdot C_{\text{in}} \cdot C_{\text{out}} \cdot H \cdot W)$ floating-point operations. Summing across $L$ layers, the total cost scales as:
	\begin{equation}
		\Gamma(f_\phi, \mathcal{I}) = \mathcal{O}\left( H \cdot W \cdot \sum_{\ell=1}^{L} k_\ell^2 \cdot C_{\text{in}}^\ell \cdot C_{\text{out}}^\ell \right) = \mathcal{O}(H \cdot W).
		\label{eq:yolo_complexity}
	\end{equation}
	The term $\sum_{\ell} k_\ell^2 C_{\text{in}}^\ell C_{\text{out}}^\ell$ depends only on the architecture (fixed by $\phi$), so computational cost is linear in the number of pixels $H \cdot W$. For YOLO11m at reference resolution $640 \times 640$, evaluating $f_\phi$ requires approximately 68 GFLOPs. Scaling the input by factor $s$ in each dimension yields an image of size $(sH) \times (sW)$ with $s^2 H W$ pixels, and cost $\Gamma(f_\phi, \mathcal{I}_s) \approx s^2 \cdot \Gamma(f_\phi, \mathcal{I})$.

	This linear dependence on pixel count enables efficient spatial partitioning. Decompose $\mathcal{I}$ into $N$ non-overlapping regions $\{\mathcal{I}_1, \ldots, \mathcal{I}_N\}$ with $\sum_i |\mathcal{I}_i| = |\mathcal{I}|$. Each region $\mathcal{I}_i$ can be processed independently: $f_\phi(\mathcal{I}_i)$ produces detections within that spatial extent, and the union $\bigcup_i f_\phi(\mathcal{I}_i)$ recovers full-frame detections (with boundary handling for objects spanning regions). The critical property is that $\Gamma(f_\phi, \mathcal{I}_i) \propto |\mathcal{I}_i|$, so distributing regions across $N$ workers reduces per-worker computation by factor $N$ while preserving total work. Fig.~\ref{fig:yolo_partitioning} illustrates this spatial partitioning approach.

	\begin{figure}[!t]
		\centering
		\begin{tabular}{c}
			\includegraphics[width=0.95\columnwidth]{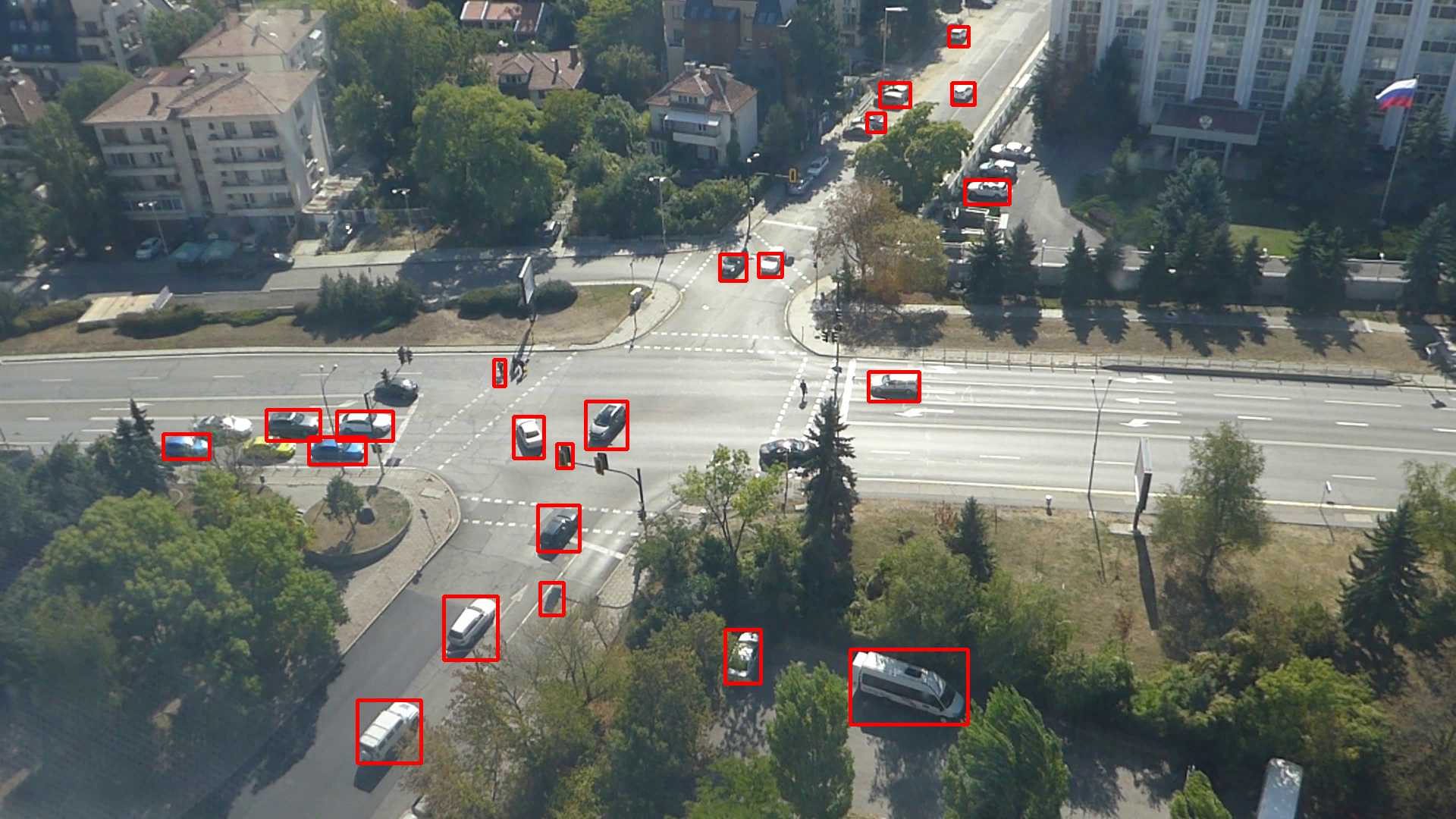} \\
			(a) \\[6pt]
			\includegraphics[width=0.95\columnwidth]{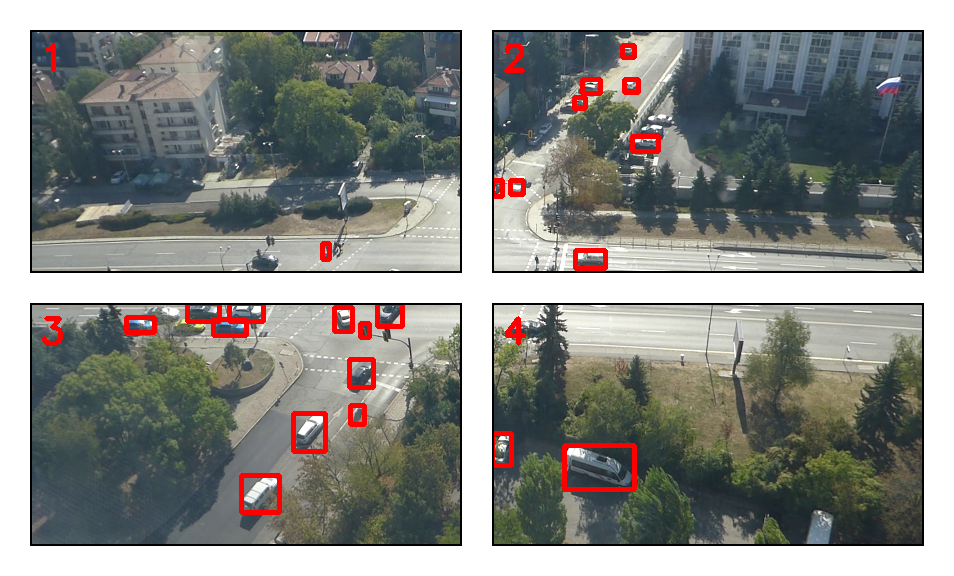} \\
			(b)
		\end{tabular}
		\caption{Spatial partitioning for distributed object detection. (a) Full frame processed on a single device. (b) Same frame partitioned into four quadrants, each processed independently by a separate worker.}
		\label{fig:yolo_partitioning}
	\end{figure}

	We now analyze the latency and throughput implications of distributed processing. For a single device with computational capacity $C$ (in GFLOPS) processing a full frame requiring $\Gamma$ GFLOPS, the inference latency is $\tau_{\text{single}} = \Gamma / C$ seconds, yielding throughput $T_{\text{single}} = C / \Gamma$ frames per second.

	Consider distributing the frame across $N$ workers. Let worker $i$ receive region $\mathcal{I}_i$ with $|\mathcal{I}_i| = |\mathcal{I}|/N$ pixels (assuming equal partitioning), requiring $\Gamma_i = \Gamma/N$ GFLOPS. Let $\tau_{\text{comm}}$ denote the one-way communication latency to transmit a region. The per-worker latency comprises transmission and computation: $\tau_i = \tau_{\text{comm}} + \Gamma_i / C_i$. Since all workers operate in parallel, the frame completes when the slowest worker finishes:
	\begin{equation}
		\tau_{\text{distributed}} = \max_{i \in \{1,\ldots,N\}} \left( \tau_{\text{comm}} + \frac{\Gamma/N}{C_i} \right).
		\label{eq:distributed_latency}
	\end{equation}
	For homogeneous workers with capacity $C_i = C$, this simplifies to $\tau_{\text{distributed}} = \tau_{\text{comm}} + \Gamma/(NC)$. The distributed throughput is:
	\begin{equation}
		T_{\text{distributed}} = \frac{1}{\tau_{\text{distributed}}} = \frac{1}{\tau_{\text{comm}} + \Gamma/(NC)} .
		\label{eq:distributed_throughput}
	\end{equation}
	When communication latency is negligible relative to computation ($\tau_{\text{comm}} \ll \Gamma/(NC)$), this reduces to $T_{\text{distributed}} \approx NC/\Gamma = N \cdot T_{\text{single}}$. The speedup factor approaches $N$ as the computation-to-communication ratio increases. The asymmetry between computation and communication costs makes this partitioning advantageous: transmitting a compressed region over local wireless requires 5-20 ms, while CPU inference requires hundreds of milliseconds, satisfying the condition $\tau_{\text{comm}} \ll \Gamma/(NC)$ in practice.

	This architecture constitutes a \textit{series configuration} (Section~\ref{subsec:series}): each worker processes a distinct region, and all must complete for the frame to be fully processed. System reliability equals the product $R_{\text{series}}(t) = \prod_{i=1}^{N} R_i(t, \Theta_i)$, where each $R_i(t, \Theta_i)$ quantifies the probability that worker $i$ meets its regional QoS constraint. The orchestrator selects workers using the analytical expressions from Sections~\ref{sec:modeling_streaming} and \ref{sec:system_reliability}.

	\section{Experiments and Simulation Results}
	\label{sec:sim_exp}

	To validate the analytical framework, we construct a software-based simulation environment emulating XED behavior under controlled conditions. The simulation captures key dynamics of computational reliability: fluctuating device capacity and variable service demand, while providing ground-truth measurements for comparison against analytical predictions.

	\subsection{Simulation Architecture}

	The simulation environment virtualizes XEDs using Docker containers, each representing an independent worker with configurable computational resources. Docker provides process-level isolation and precise control over CPU allocation, enabling us to emulate devices with different computational capabilities on a single host machine. Each container runs the YOLO11m object detection model, processing video frames as a representative streaming inference workload.

	The key insight enabling this simulation is that computational capacity and service demand can be independently controlled through two mechanisms: \textit{thread allocation} governs capacity, while \textit{frame scale} governs demand. By varying these parameters dynamically during video processing, we generate the stochastic capacity and demand distributions that the reliability framework models.

	\textbf{Capacity through thread allocation.} A container's computational capacity $C_i(t)$ is determined by the number of CPU threads allocated to inference at time $t$. Modern CPUs execute floating-point operations in parallel across threads; more threads yield higher throughput, measured in GFLOPS. We profile the YOLO11m model across thread counts from 2 to 12, establishing a mapping from thread count to achieved GFLOPS. During simulation, the orchestrator dynamically adjusts thread allocation, emulating how a real XED's available capacity fluctuates as background applications consume or release CPU resources. A container configured with thread range $[n_{\min}, n_{\max}]$ experiences capacity that varies within corresponding GFLOPS bounds $[C_i^{\min}, C_i^{\max}]$.

	\textbf{Demand through frame scale.} Service demand $\Delta_i(t)$ is controlled by the spatial resolution of input frames. As established in Section~\ref{sec:case_study}, YOLO's computational cost scales quadratically with frame dimensions. A \textit{scale factor} $s \in (0, 1]$ reduces both width and height proportionally, yielding a frame with $s^2$ times the pixels and approximately $s^2$ times the computational demand. For instance, scale $s = 0.5$ produces a frame with 25\% of the original pixels, requiring roughly 25\% of full-frame GFLOPS. By varying the scale factor during simulation, we generate demand sequences that span the range $[\Delta_i^{\min}, \Delta_i^{\max}]$ corresponding to scale bounds $[s_{\min}, s_{\max}]$.

	\textbf{QoS as frame rate.} The QoS threshold $\Theta$ translates directly to a frame rate requirement in the object detection context. If the application requires $f$ frames per second (FPS), each frame must complete within $1/f$ seconds. Given demand $\Delta_i$ GFLOPS per frame, sustaining $f$ FPS requires capacity $C_i \geq f \cdot \Delta_i$ GFLOPS, or equivalently $C_i/\Delta_i \geq f$. Thus, the QoS threshold $\Theta$ corresponds to the target FPS: $\Theta = f_{\text{target}}$. A frame processed successfully (capacity exceeds $\Theta$ times demand) meets the real-time deadline; a frame that fails this condition misses its deadline and degrades user experience.
\begin{figure*}[!th]
		\centering
		\begin{tabular}{ccc}
			\includegraphics[width=0.3\textwidth]{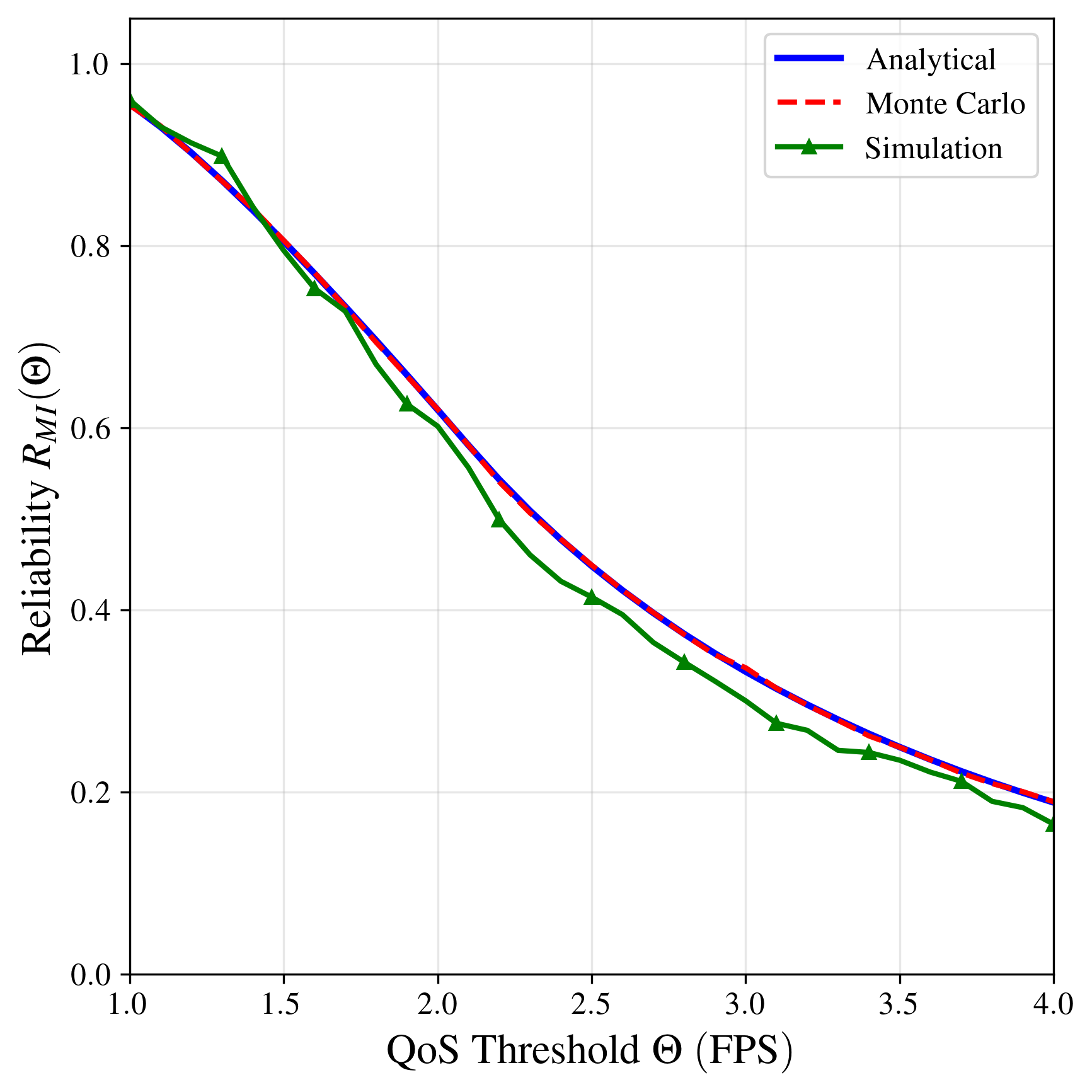} &
			\includegraphics[width=0.3\textwidth]{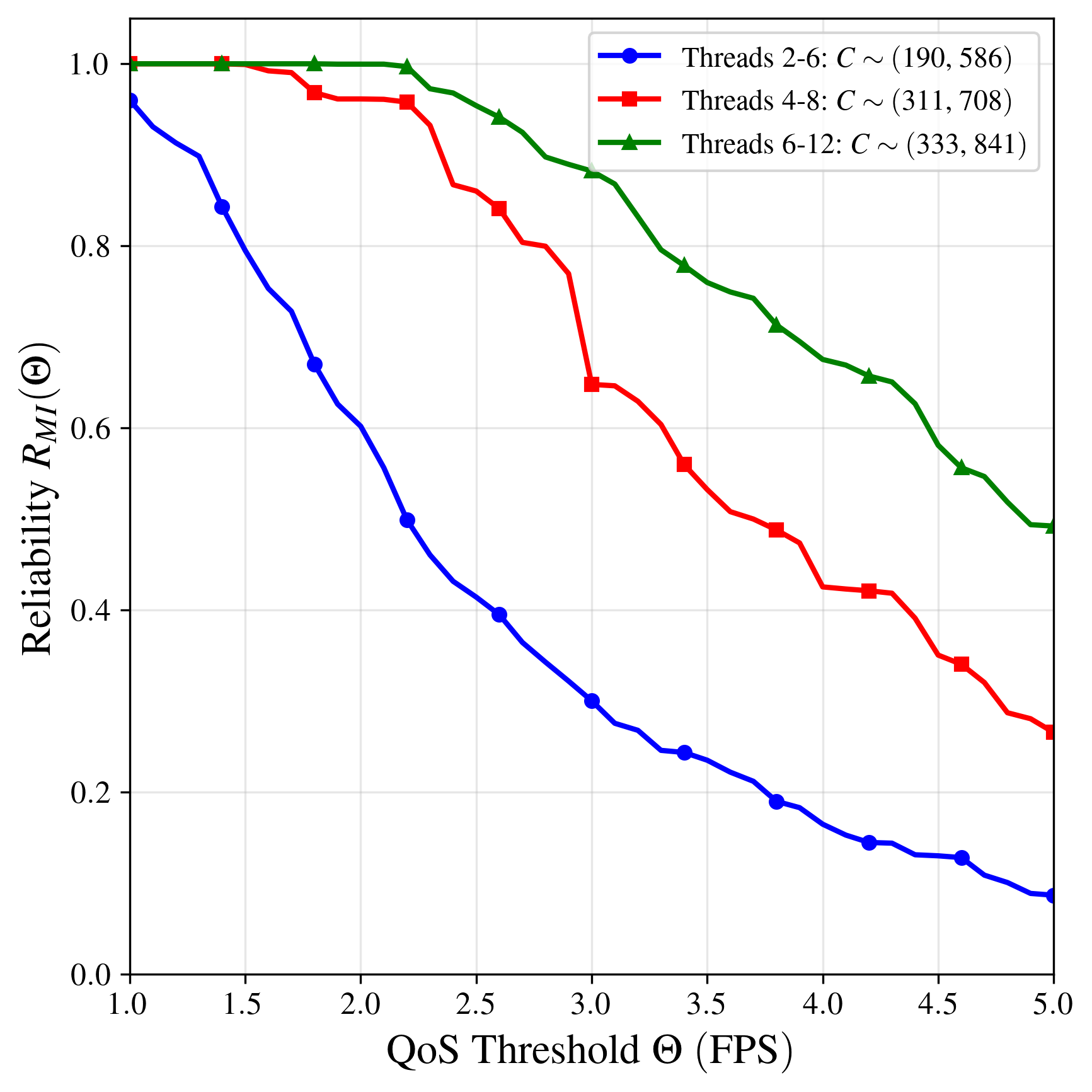} &
			\includegraphics[width=0.3\textwidth]{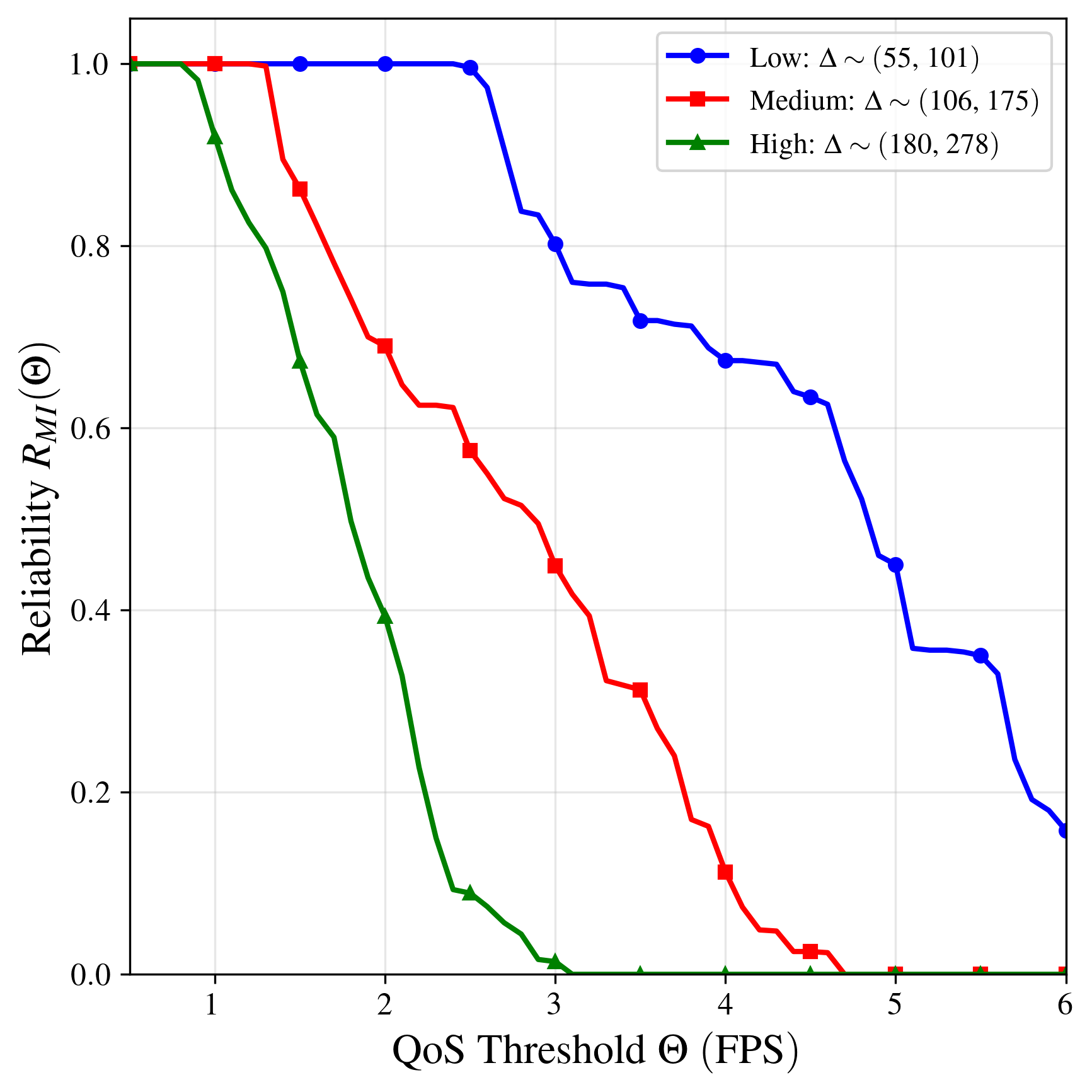} \\
			(a) & (b) & (c)
		\end{tabular}
		\caption{Validation of $R^{\text{MI}}_i(t, \Theta)$. (a) Analytical vs Monte Carlo vs simulation. (b) Effect of capacity. (c) Effect of demand.}
		\label{fig:mi_validation}
	\end{figure*}
	\subsection{Profiling and Calibration}

	Prior to reliability experiments, we profile the YOLO11m model exhaustively across all thread-scale combinations. For each configuration (threads $n \in \{2, \ldots, 12\}$, scale $s \in \{0.1, 0.2, \ldots, 0.9\}$), we measure inference time over multiple runs to obtain stable estimates. This profiling yields two critical mappings: (i) thread count to capacity in GFLOPS, and (ii) scale factor to demand in GFLOPS. These mappings allow the simulation to translate thread and scale configurations into the $(C_i, \Delta_i)$ pairs that the analytical framework requires.

	The profiling reveals expected scaling behavior. Capacity increases approximately linearly with thread count up to the physical core limit, with diminishing returns as threads compete for shared resources. Demand scales quadratically with scale factor, consistent with the convolutional structure of the YOLO architecture. At full scale ($s = 0.9$) with minimal threads ($n = 2$), inference requires over one second per frame (sub-1 FPS). At reduced scale ($s = 0.25$) with maximal threads ($n = 12$), inference completes in under 40 ms (exceeding 25 FPS). This dynamic range, spanning nearly two orders of magnitude in throughput, provides a rich space for reliability analysis.

	\subsection{Experimental Procedure}

	Each experiment processes a video stream while dynamically varying thread allocation and frame scale according to specified distributions. At regular intervals (every 30 or 20 frames), the simulation samples new values for threads and scale from configured ranges, emulating the stochastic fluctuations of real XED environments. For each frame, we record: (i) the thread count (capacity proxy), (ii) the scale factor (demand proxy), (iii) inference time, and (iv) whether the frame met the QoS threshold.

	The empirical reliability is computed as the fraction of frames meeting the QoS constraint over the experiment duration. We compare this empirical value against the analytical predictions from the MI framework (using only the declared bounds) and the historical data framework (using MLE-estimated parameters from observations). This comparison validates both the accuracy of the closed-form expressions and the benefit of incorporating historical observations.
	
	\subsection{Simulation Results}

	Fig.~\ref{fig:mi_validation}(a) validates the $R^{\text{MI}}_i(t, \Theta)$ closed-form expression (Eq.~\ref{eq:MI_Theta_reliability_simplified_formatted}) against Monte Carlo sampling and empirical simulation. The experiment processes 2590 video frames with thread allocation uniformly sampled from $[2,6]$ and frame scale from $[0.4,0.9]$, yielding capacity bounds $C_i\in[55,152]$ GFLOPS and demand bounds $\Delta_i\in[55,278]$ GFLOPS. Monte Carlo evaluation draws $10^5$ independent $(C_i,\Delta_i)$ pairs from the uniform distributions and computes $P(C_i/\Delta_i \geq \Theta)$. The simulation curve reports the empirical fraction of frames achieving each QoS threshold. All three curves align across the threshold range $\Theta\in[1,4]$ FPS, confirming the analytical derivation.

	Fig.~\ref{fig:mi_validation}(b) examines the effect of computational capacity on reliability. Three configurations vary thread allocation while holding the scale range fixed at $[0.4,0.9]$: threads $\in[2,6]$ (mean FPS 2.66), $[4,8]$ (mean FPS 4.20), and $[6,12]$ (mean FPS 5.47). Higher thread counts increase capacity bounds, shifting the reliability curve rightward. A device with threads $\in[6,12]$ maintains $R>0.9$ up to $\Theta\approx3$ FPS, whereas threads $\in[2,6]$ drops below 0.9 at $\Theta\approx1.5$ FPS.

	Fig.~\ref{fig:mi_validation}(c) isolates the effect of demand by partitioning the same experiment (threads $\in[2,6]$) into three demand regimes based on frame scale. Low demand ($\Delta_i<104$ GFLOPS, scale $<0.55$) achieves near-unity reliability up to $\Theta\approx4$ FPS. High demand ($\Delta_i>178$ GFLOPS, scale $>0.72$) degrades reliability rapidly, falling below 0.5 at $\Theta\approx2$ FPS. This confirms the intuition that reliability depends on the ratio $C_i/\Delta_i$: reducing demand has an equivalent effect to increasing capacity.

	\begin{figure*}[!t]
		\centering
		\begin{tabular}{cccc}
			\includegraphics[width=0.22\textwidth]{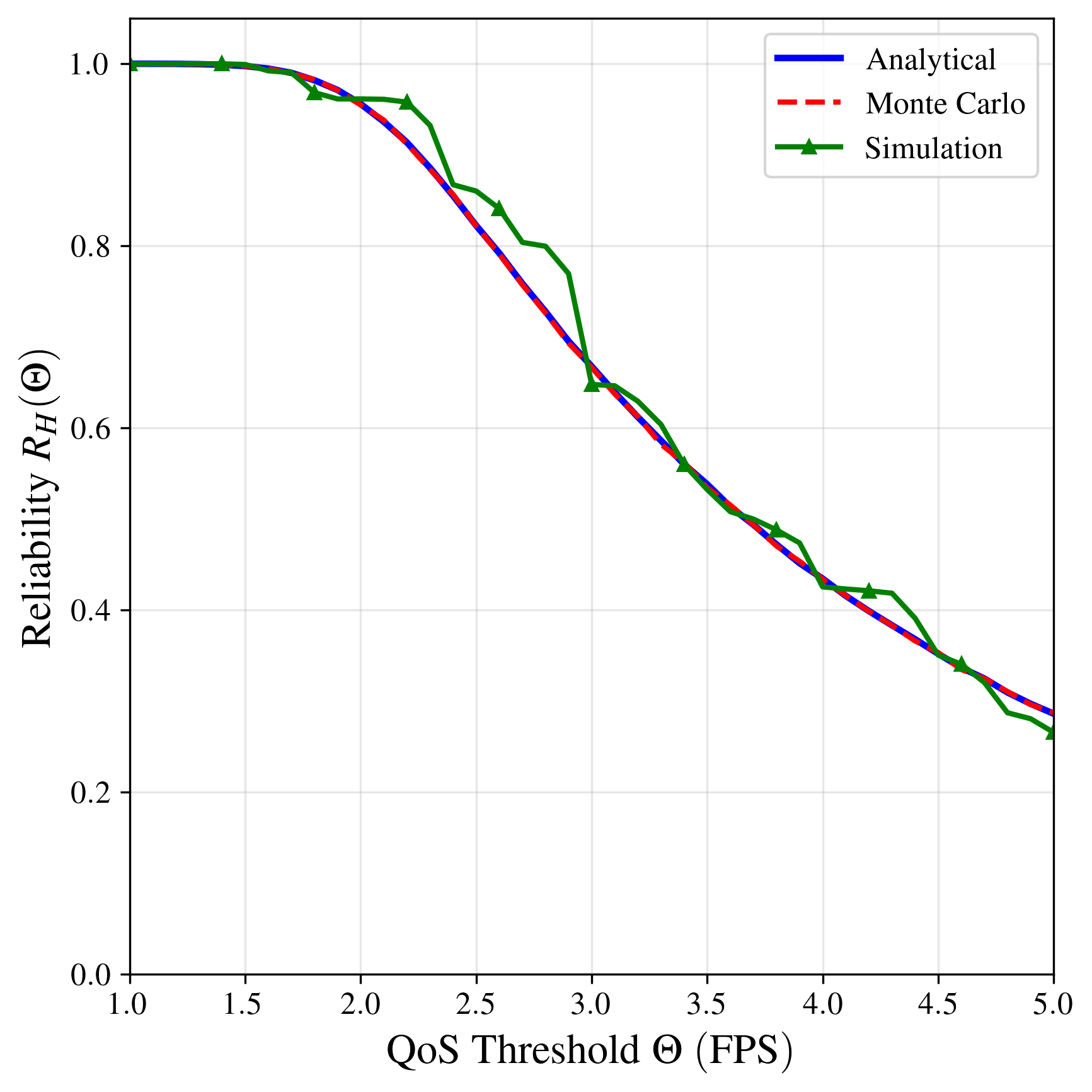} &
			\includegraphics[width=0.22\textwidth]{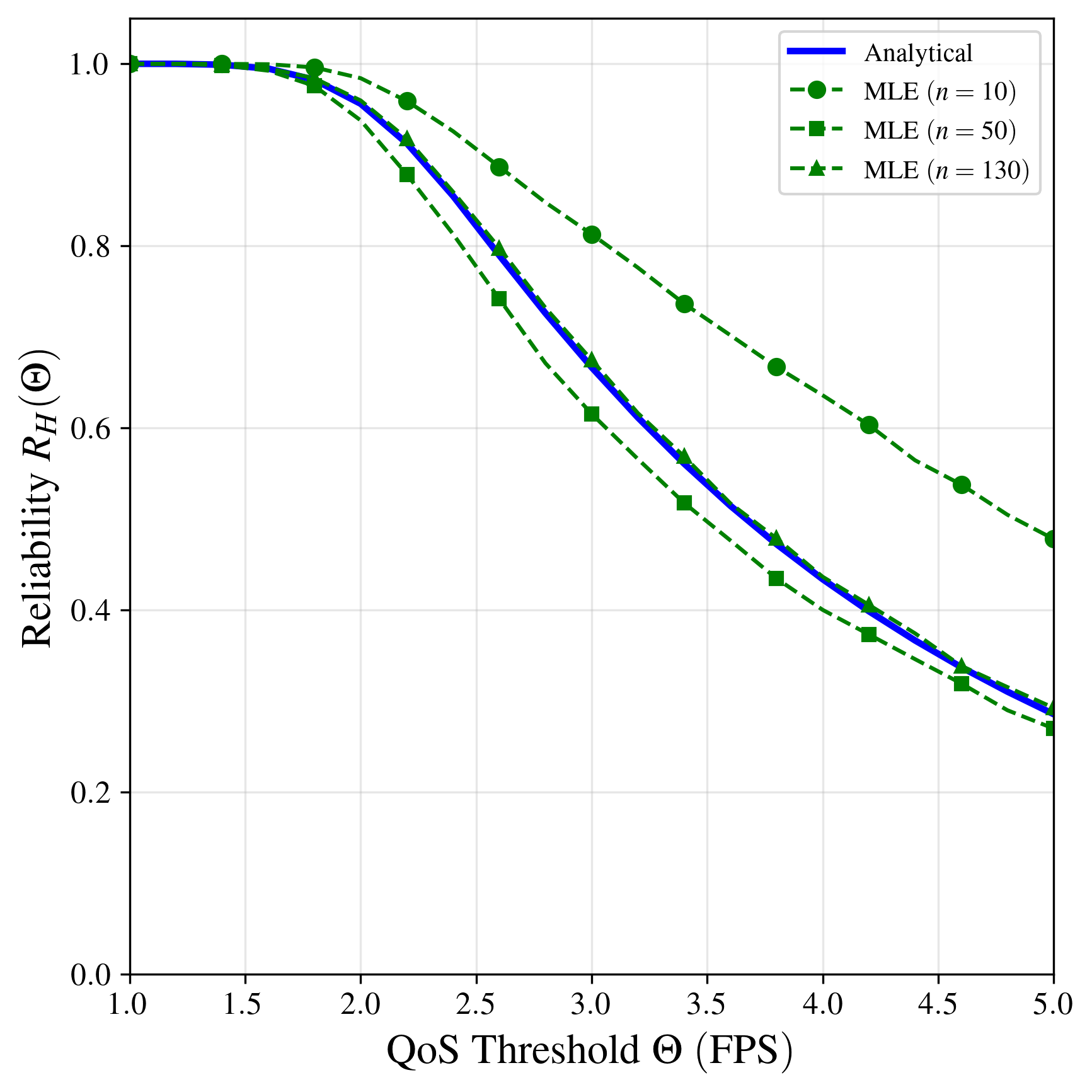} &
			\includegraphics[width=0.22\textwidth]{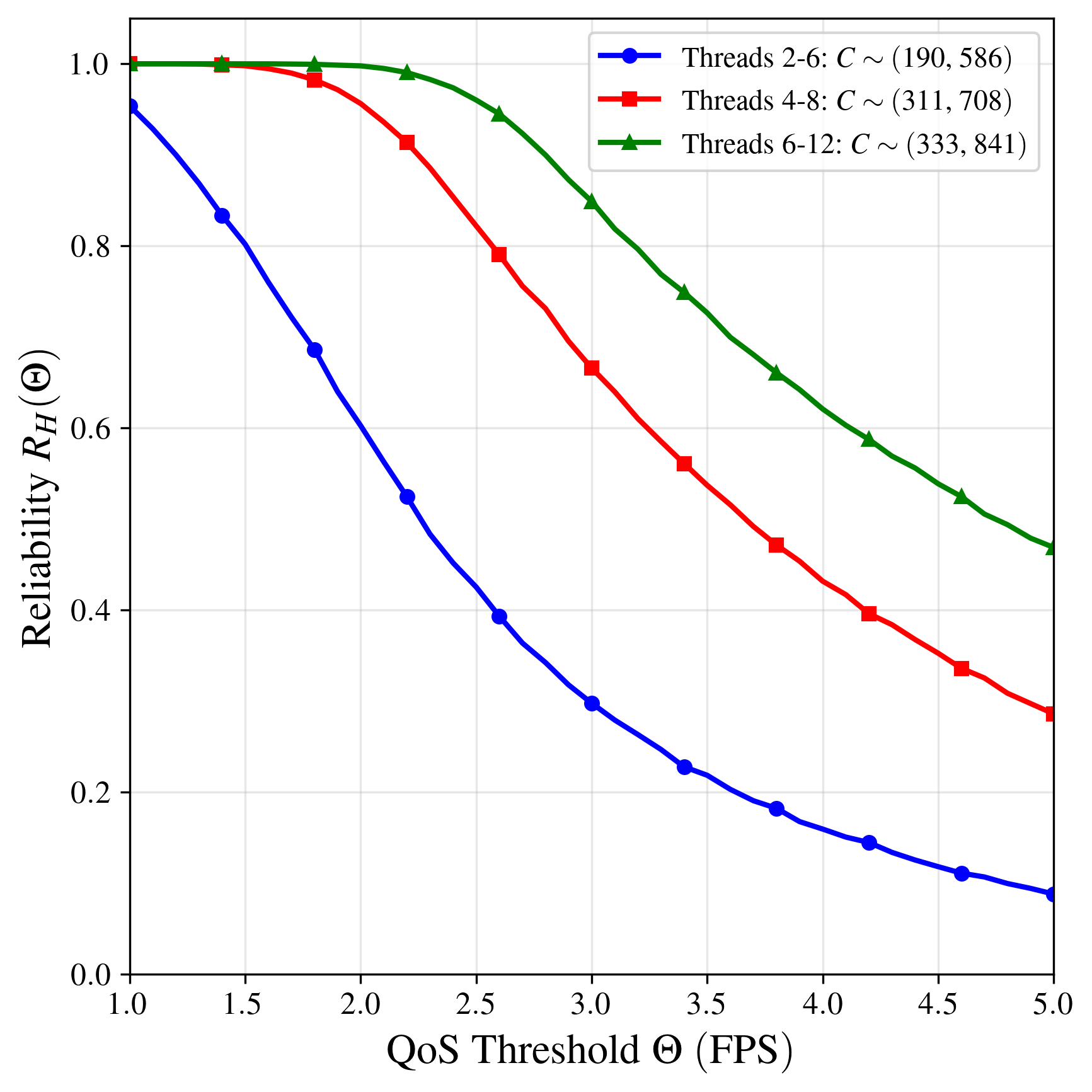} &
			\includegraphics[width=0.22\textwidth]{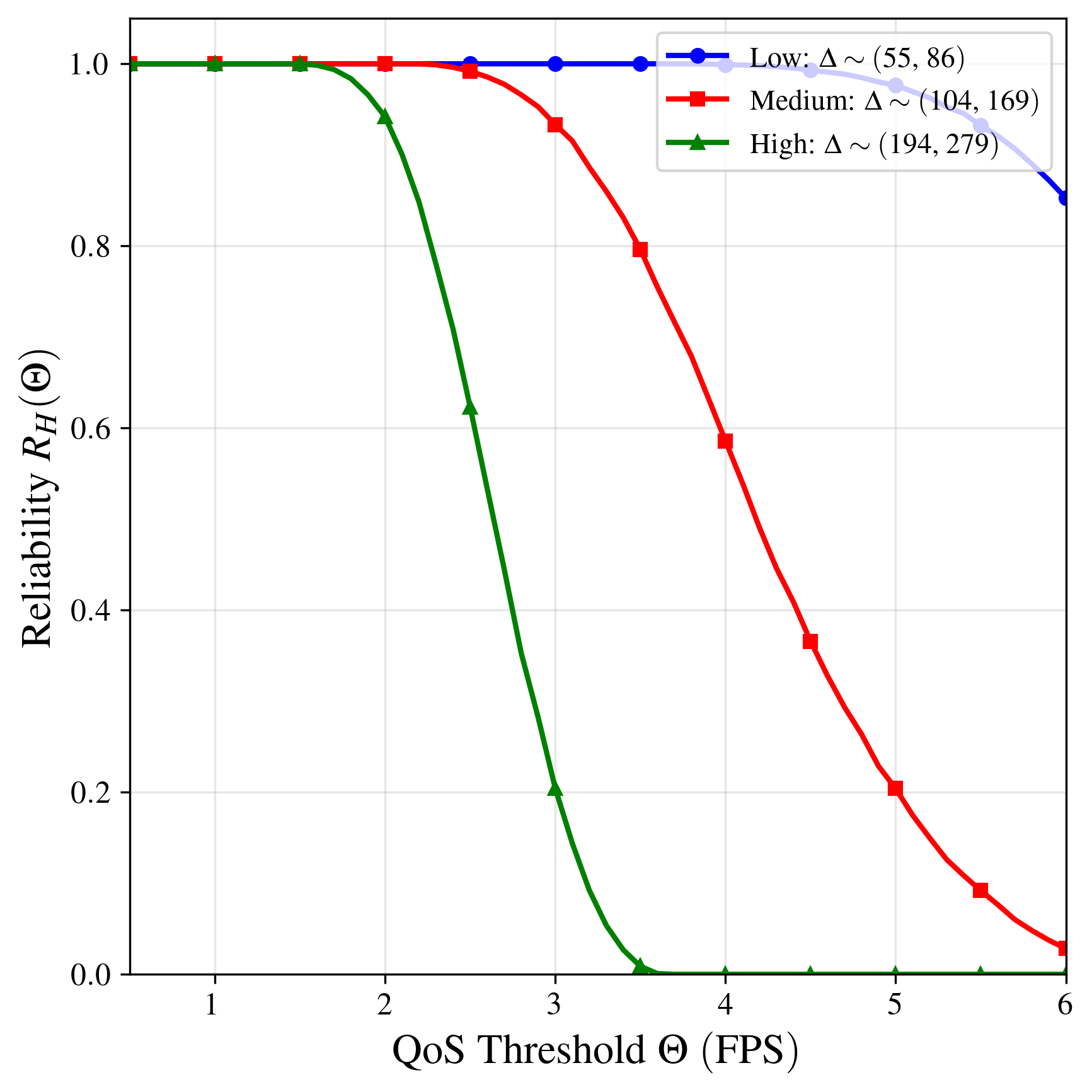} \\
			(a) & (b) & (c) & (d)
		\end{tabular}
		\caption{Validation of $R^{\text{H}}_i(t, \Theta)$. (a) Analytical vs MC vs simulation. (b) MLE convergence with sample size. (c) Effect of capacity. (d) Effect of demand.}
		\label{fig:rh_validation}
	\end{figure*}

	Fig.~\ref{fig:rh_validation} validates the historical data framework $R^{\text{H}}_i(t, \Theta)$ (Lemma~2), demonstrating that MLE-refined estimates yield tighter predictions than the bounds-only $R^{\text{MI}}_i(t, \Theta)$ approach.

	Fig.~\ref{fig:rh_validation}(a) compares the analytical $R^{\text{H}}_i(t, \Theta)$ expression against Monte Carlo sampling and empirical simulation for threads $\in[4,8]$ and scale $\in[0.4,0.9]$. MLE fits truncated normal distributions to 2590 observed $(C_i,\Delta_i)$ pairs, yielding parameters $(\mu_{C_i},\sigma_{C_i})$ and $(\mu_{\Delta_i},\sigma_{\Delta_i})$ that capture the empirical concentration around the mean. The three curves align closely, confirming the integral formulation in Eq.~\ref{eq:H_Theta_reliability_presented_earlier_lemma_formatted}.

	Fig.~\ref{fig:rh_validation}(b) examines how MLE estimates converge as sample size increases. With only $n=10$ independent samples, the estimated reliability deviates visibly from the ground truth. At $n=50$, the estimate improves substantially, and at $n=130$ (the full dataset subsampled at the change interval), it matches the analytical curve. This demonstrates that orchestrators can progressively refine reliability predictions as interaction history accumulates.

	Fig.~\ref{fig:rh_validation}(c) and (d) mirror the capacity and demand sensitivity analyses from Fig.~\ref{fig:mi_validation}, now using the $R^{\text{H}}_i(t, \Theta)$ framework. The qualitative behavior is consistent: higher capacity shifts the curve rightward, lower demand shifts it rightward. However, $R^{\text{H}}_i$ predictions are sharper because MLE captures the actual distribution shape rather than assuming uniformity over the full bounds. For instance, at $\Theta=3$ FPS with threads $\in[4,8]$, $R^{\text{H}}_i$ predicts reliability $\approx0.85$ compared to $R^{\text{MI}}_i\approx0.69$, closer to the empirical value of $0.87$.

	\begin{figure*}[!t]
		\centering
		\begin{tabular}{ccc}
			\includegraphics[width=0.3\textwidth]{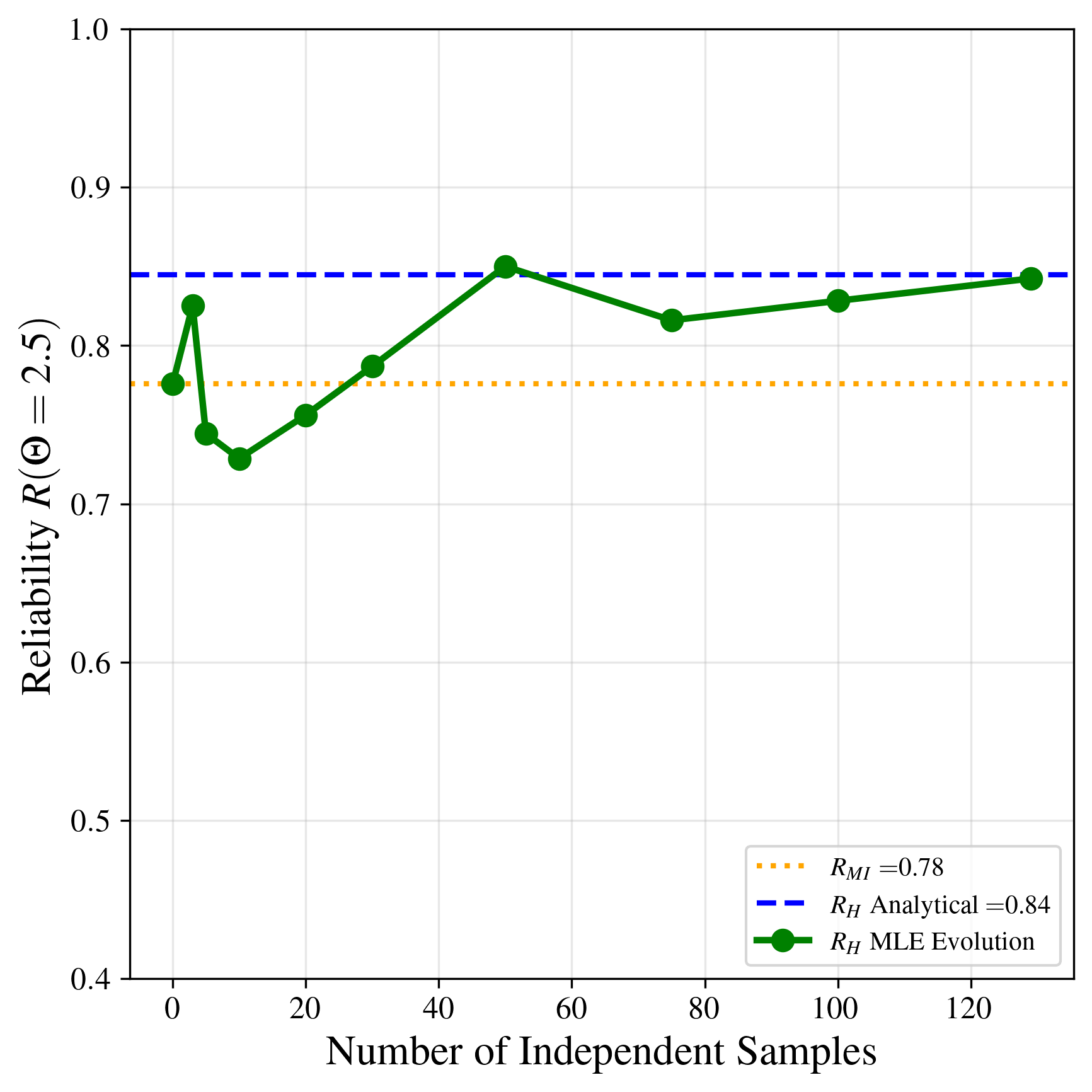} &
			\includegraphics[width=0.3\textwidth]{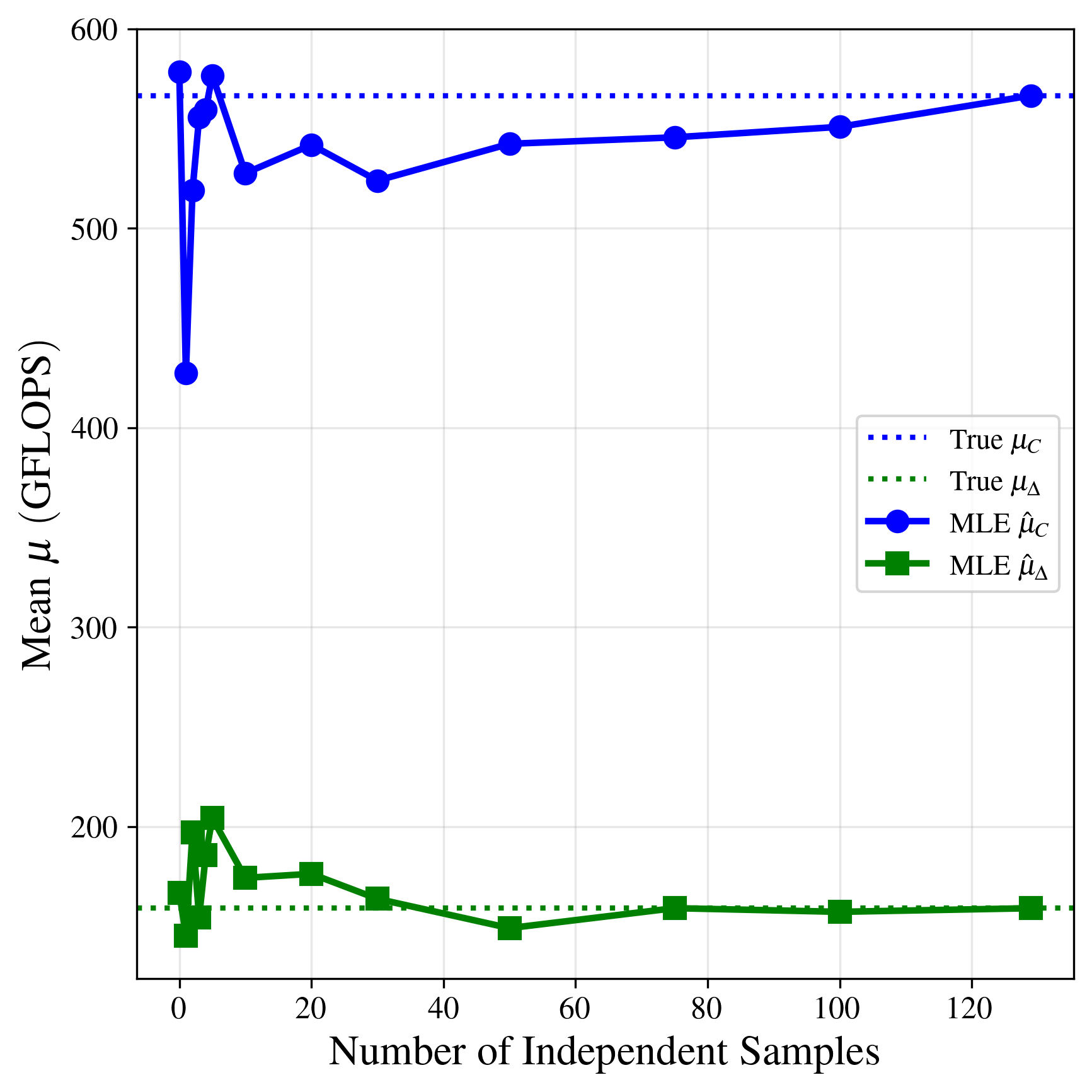} &
			\includegraphics[width=0.3\textwidth]{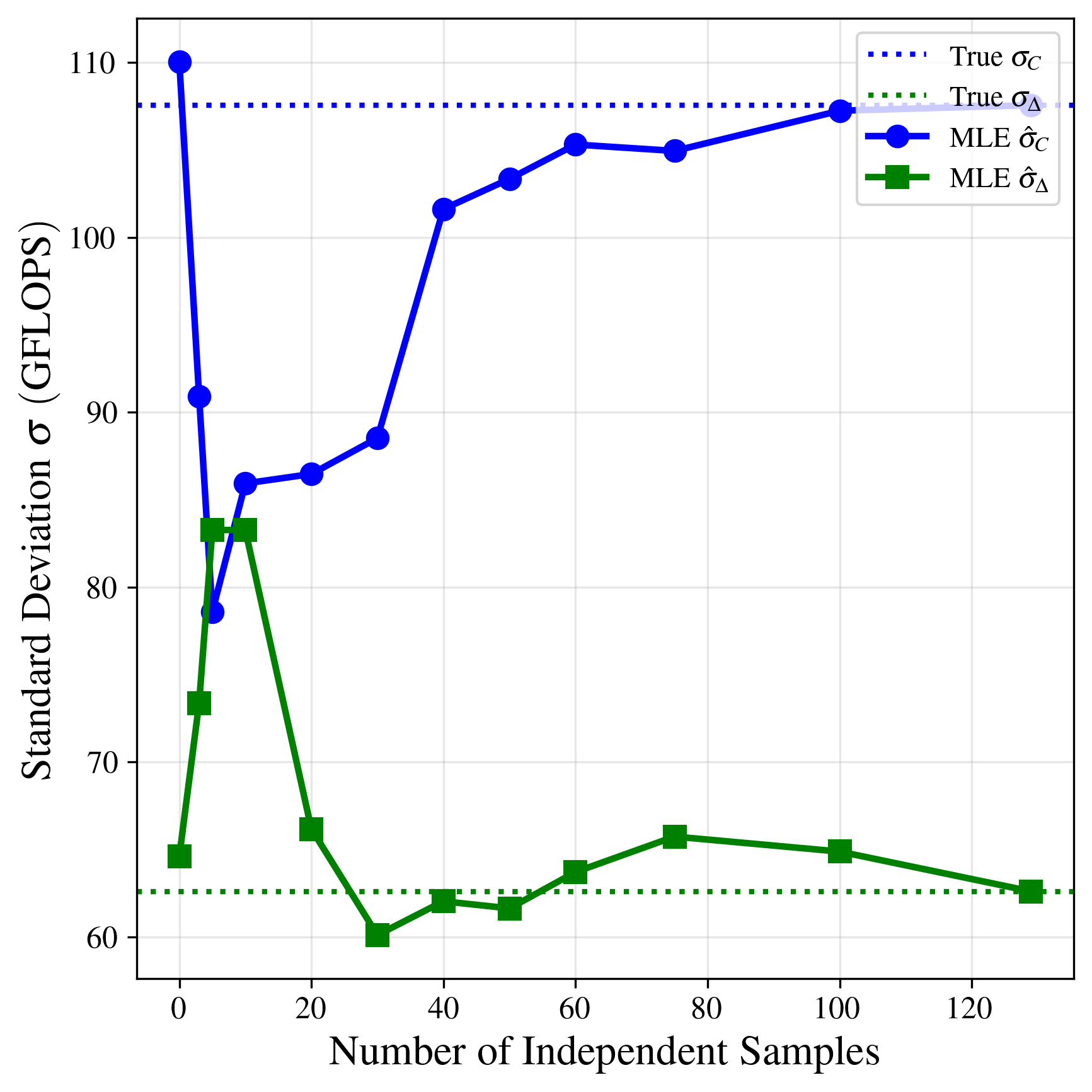} \\
			(a) & (b) & (c)
		\end{tabular}
		\caption{MLE parameter convergence. (a) Reliability estimate evolution. (b) Mean convergence. (c) Standard deviation convergence.}
		\label{fig:mle_convergence}
	\end{figure*}

	Fig.~\ref{fig:mle_convergence} illustrates the progressive refinement of $R^{\text{H}}_i(t, \Theta)$ estimates as historical observations accumulate, using an experiment with threads $\in[4,12]$ and scale $\in[0.4,0.9]$.

	Fig.~\ref{fig:mle_convergence}(a) tracks reliability at $\Theta=2.5$ FPS as a function of independent sample count. With zero samples, the orchestrator defaults to $R^{\text{MI}}_i\approx0.55$ based on uniform assumptions. As samples accumulate, the MLE-based $R^{\text{H}}_i$ estimate rises and converges to the analytical value $R^{\text{H}}_i\approx0.82$ at approximately $n=50$ independent observations. This confirms the practical benefit of historical data: reliability estimates improve by $\approx0.27$ absolute over the bounds-only baseline.

	Fig.~\ref{fig:mle_convergence}(b) and (c) decompose this convergence into the underlying MLE parameters. The mean estimates $\hat{\mu}_{C_i}$ and $\hat{\mu}_{\Delta_i}$ stabilize within 10-20 samples, while the standard deviation estimates $\hat{\sigma}_{C_i}$ and $\hat{\sigma}_{\Delta_i}$ require 30-50 samples to converge. The true parameters (dotted lines) are recovered from the full dataset of 130 independent samples. Starting from uniform assumptions (midpoint of bounds for $\mu$, $(\text{range})/\sqrt{12}$ for $\sigma$), the MLE estimates converge monotonically toward the true values.

	\begin{figure}[!t]
		\centering
		\begin{tabular}{cc}
			\includegraphics[width=0.45\columnwidth]{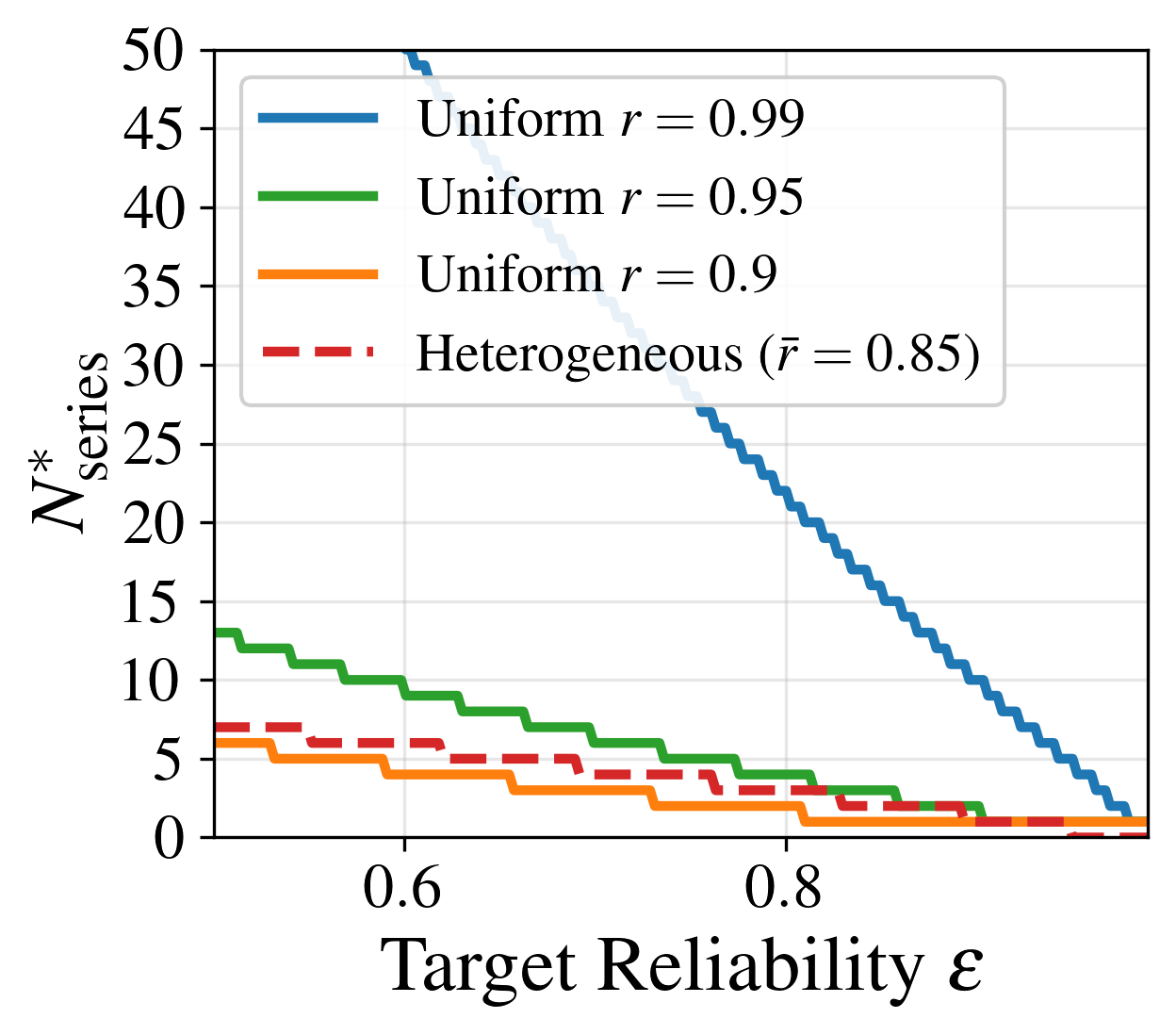} &
			\includegraphics[width=0.45\columnwidth]{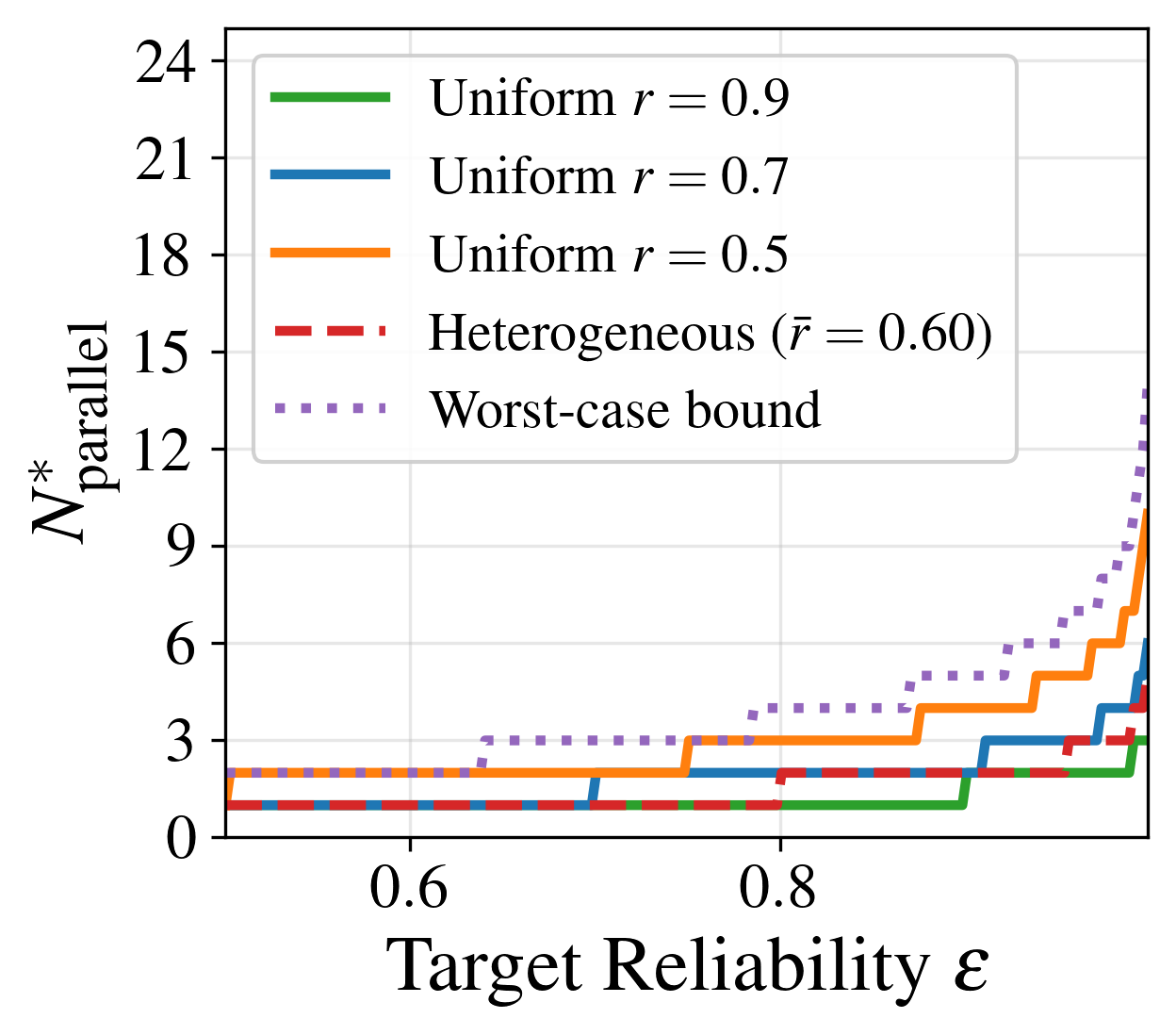} \\
			(a) & (b)
		\end{tabular}
		\caption{Device selection bounds. (a) Series: max devices $N^*_{\mathrm{series}}$. (b) Parallel: min devices $N^*_{\mathrm{parallel}}$.}
		\label{fig:device_selection}
	\end{figure}

	Fig.~\ref{fig:device_selection} validates the device selection bounds from Lemmas 4 and 5, showing the maximum feasible series devices and minimum required parallel devices as functions of target reliability $\varepsilon$.

	Fig.~\ref{fig:device_selection}(a) plots $N^*_{\mathrm{series}}$ for uniform device pools with $R\in\{0.90, 0.95, 0.99\}$ and a heterogeneous pool ($\bar{R}=0.85$, 20 devices with $R\in[0.75,0.95]$). At $\varepsilon=0.9$, a uniform pool with $R=0.99$ supports up to 10 series devices, while $R=0.95$ supports only 2. The heterogeneous case tracks between uniform bounds, confirming that selecting the most reliable devices first maximizes the series configuration size.

	Fig.~\ref{fig:device_selection}(b) plots $N^*_{\mathrm{parallel}}$ for uniform pools with $R\in\{0.50, 0.70, 0.90\}$, a heterogeneous pool ($\bar{R}=0.60$), and the worst-case bound from Eq.~\ref{eq:parallel_upper_bound}. At $\varepsilon=0.99$, achieving high reliability requires 2 devices when $R=0.90$, but 9 devices when $R=0.50$. The heterogeneous case requires fewer devices than the worst-case bound, validating that prioritizing high-reliability devices reduces redundancy requirements.

	\begin{figure*}[!t]
		\centering
		\begin{tabular}{ccc}
			\includegraphics[width=0.3\textwidth]{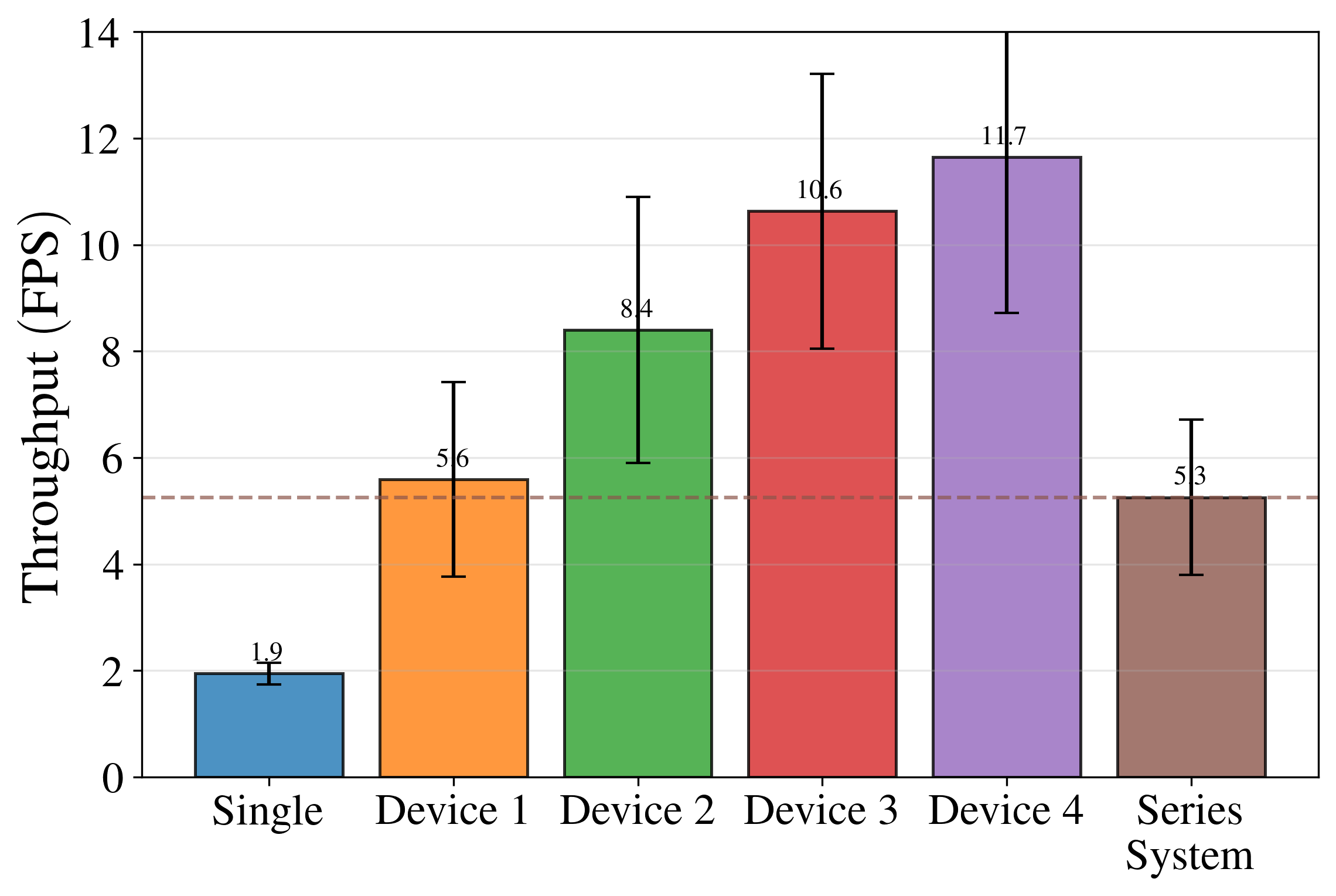} &
			\includegraphics[width=0.3\textwidth]{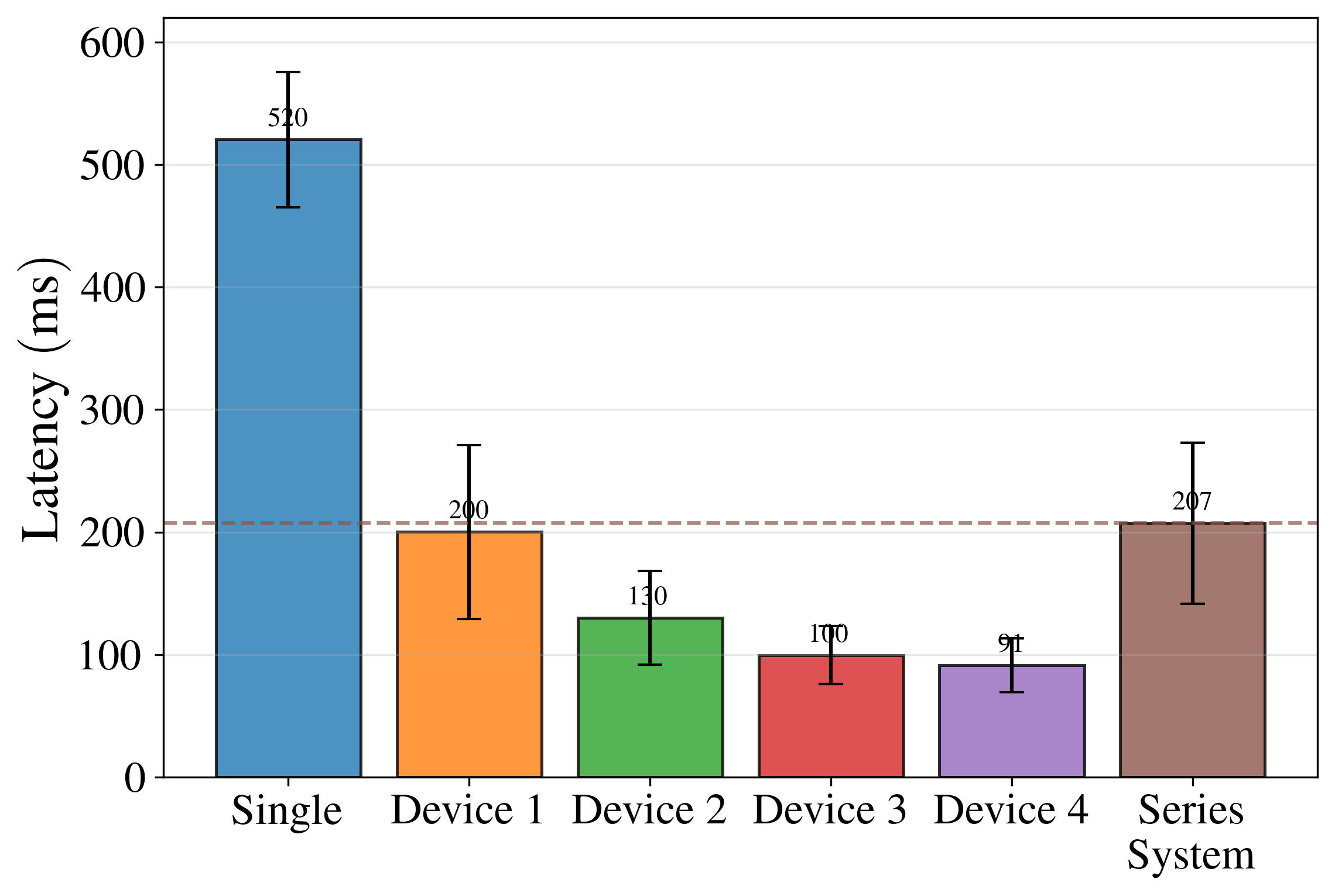} &
			\includegraphics[width=0.3\textwidth]{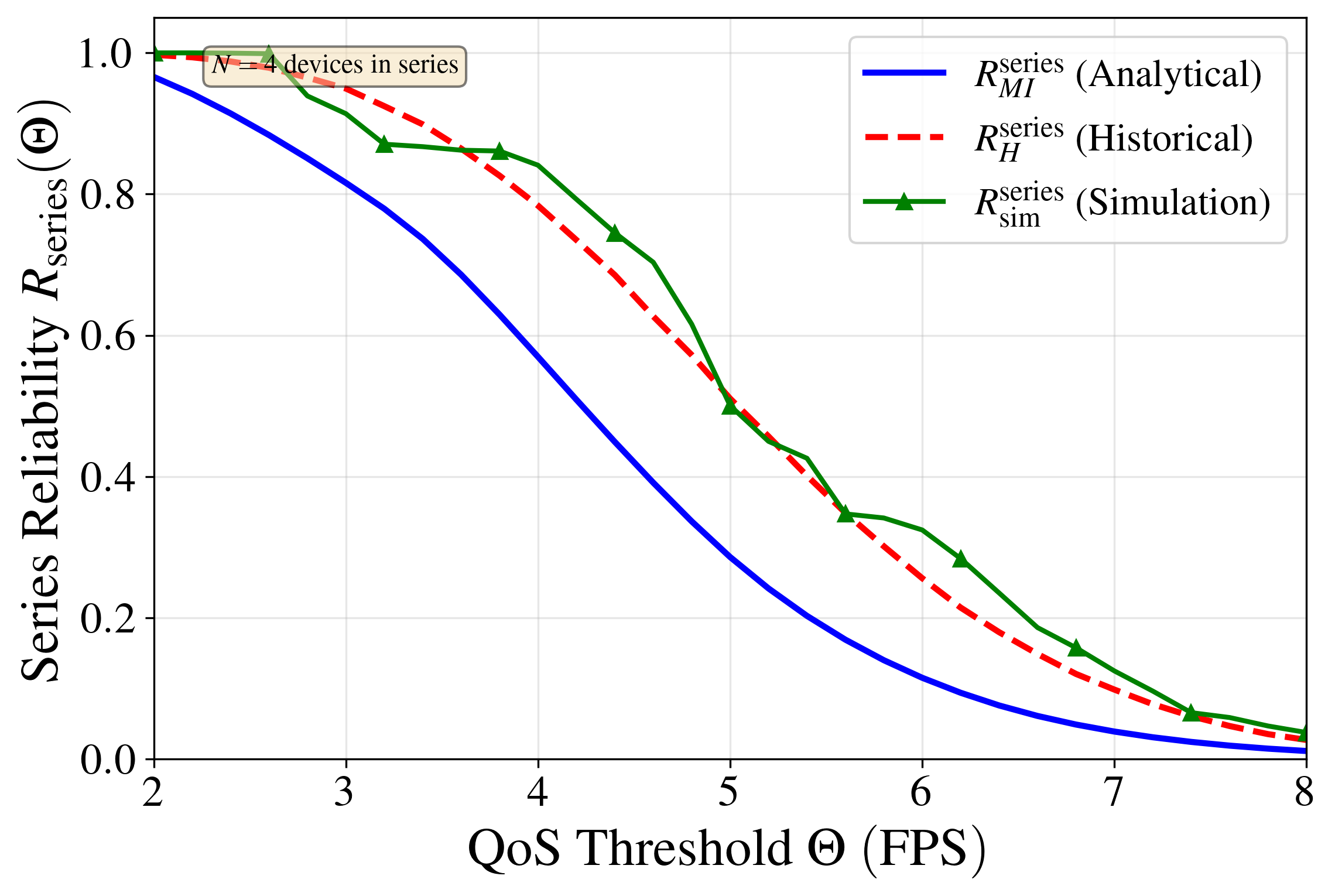} \\
			(a) & (b) & (c)
		\end{tabular}
		\caption{Series system with spatial partitioning. (a) Throughput comparison. (b) Latency comparison. (c) Series reliability validation.}
		\label{fig:series_system}
	\end{figure*}

	Fig.~\ref{fig:series_system} validates the series configuration analysis using a 4-device spatial partitioning deployment for YOLO11m inference. A single device processing full-resolution frames with threads $\in[6,10]$ achieves 1.9 FPS at 520 ms latency. In the series configuration, the frame is partitioned across 4 devices with heterogeneous thread allocations (threads $\in[2,4]$, $[4,6]$, $[6,8]$, $[8,10]$ for Devices 1-4 respectively), each processing 1/4 of the frame.

	Fig.~\ref{fig:series_system}(a) shows individual device throughputs ranging from 5.6 FPS (Device 1) to 11.7 FPS (Device 4). The series system throughput equals the minimum across devices, yielding 5.3 FPS, a $2.8\times$ improvement over single-device processing. Fig.~\ref{fig:series_system}(b) shows corresponding latencies: individual devices achieve 91-200 ms per region, while the series system latency (the maximum across parallel executions) is 207 ms, a $2.5\times$ reduction from the single-device baseline.

	Fig.~\ref{fig:series_system}(c) validates the series reliability formula (Eq.~\ref{eq:series_reliability}). The analytical $R^{\text{MI}}_{\text{series}}(t)$ computed as the product of individual $R^{\text{MI}}_i(t, \Theta)$ values provides a conservative lower bound. The historical $R^{\text{H}}_{\text{series}}(t)$ using MLE-fitted parameters aligns closely with simulation results across the QoS range $\Theta\in[2,8]$ FPS. The multiplicative penalty of series configurations is evident: at $\Theta=5$ FPS, individual device reliabilities exceed 0.8, but the series reliability drops to approximately 0.5 due to the product $\prod_{i=1}^{4} R_i(t, \Theta)$.

	\section{Conclusion}
	\label{sec:conc}
	This paper introduced a formal analytical framework to quantify the computational reliability of AI inference streaming services in \gls{XEC} environments. We defined reliability as the probability of an \gls{XED} meeting a specified \gls{QoS} threshold and developed models for two scenarios: minimal information, where no prior knowledge of worker behavior exists, and historical data, where past observations yield refined estimates. The framework extends from single-device analysis to system-level configurations, providing closed-form expressions for series, parallel, and distributed workload systems. We derived optimal workload partitioning rules that equalize marginal log-reliability across devices, along with analytical bounds for device selection to achieve target system reliability. Experimental validation using real YOLO11m inference under dynamic resource fluctuations confirmed the accuracy of our analytical expressions, with the historical data approach providing progressively refined estimates compared to the bounds-based \gls{MI} model.

	Our work equips \gls{XEC} orchestrators with analytical tools for assessing system feasibility, determining optimal device selection, and configuring workload distribution across multi-device deployments for AI inference at the extreme edge. The framework focuses on computational reliability; extending the analysis to incorporate communication latency and network reliability remains an important direction for future work.
	\section*{Acknowledgment}
	This research was supported by the Natural Sciences and Engineering Research Council of Canada (NSERC) under Grant RGPIN-2025-05001.

	\appendices

	\section{Proof of Lemma~\ref{lemma:streaming_MI_reliability_main}}
	\label{app:proof_lemma1}

	The reliability $R^{\text{MI}}_i(t, \Theta) = P(C_i(t) \geq \Theta \Delta_i(t))$ is obtained by integrating the joint PDF $f_{C_i,\Delta_i}(c,\delta) = \frac{1}{C_i^{\text{range}}\Delta_i^{\text{range}}}$ over the region where $c \geq \Theta \delta$, subject to $C_i^{\min} \leq c \leq C_i^{\max}$ and $\Delta_i^{\min} \leq \delta \leq \Delta_i^{\max}$.

	The condition $c \geq \Theta \delta$ implies $\delta \leq c/\Theta$. Since $c \leq C_i^{\max}$, we must have $\delta \leq C_i^{\max}/\Theta$. Thus, the upper limit for $\delta$ is $\min(\Delta_i^{\max}, C_i^{\max}/\Theta)$. Under the assumption that $C_i^{\min} \leq \Theta \delta$ for the relevant range of $\delta$, the lower integration limit for $c$ is $\Theta \delta$.

	Let $\delta_u = \min(\Delta_i^{\max}, C_i^{\max}/\Theta)$. Then:
	\begin{align}
		R^{\text{MI}}_i(t, \Theta) &= \frac{1}{C_i^{\text{range}}\Delta_i^{\text{range}}} \int_{\Delta_i^{\min}}^{\delta_u}  \int_{\Theta \delta}^{C_i^{\max}} dc \,d\delta \nonumber \\
		&= \frac{1}{C_i^{\text{range}}\Delta_i^{\text{range}}} \int_{\Delta_i^{\min}}^{\delta_u} (C_i^{\max} - \Theta \delta) \,d\delta.
	\end{align}

	Evaluating the integral:
	\begin{align}
		R^{\text{MI}}_i(t, \Theta) &= \frac{1}{C_i^{\text{range}}\Delta_i^{\text{range}}} \left[ C_i^{\max}\delta - \frac{\Theta \delta^2}{2} \right]_{\Delta_i^{\min}}^{\delta_u} \nonumber \\
		&= \frac{1}{C_i^{\text{range}}\Delta_i^{\text{range}}} \biggl[ \left( C_i^{\max}\delta_u - \frac{\Theta \delta_u^2}{2} \right) \nonumber \\
		&\quad -  \left( C_i^{\max}\Delta_i^{\min} - \frac{\Theta (\Delta_i^{\min})^2}{2} \right) \biggr].
	\end{align}

	This yields Eq.~\eqref{eq:MI_Theta_reliability_simplified_formatted}. The derivation relies on $C_i^{\min} \leq \Theta \delta$ for the integration range, so that the lower limit for $c$ is $\Theta \delta$.

	\section{Proof of Lemma~\ref{lemma:streaming_si_reliability_main}}
	\label{app:proof_lemma2}

	For notational brevity, we suppress the time index $t$ in the MLE parameters; all $\mu$ and $\sigma$ terms are understood to be evaluated at time $t$. Assuming $\Theta \delta \ge C_i^{\min}$:
	\begin{align}
		R^{\text{H}}_i(t, \Theta) &= \int_{\Delta_i^{\min}}^{\Delta_i^{\max}} \int_{\Theta \delta}^{C_i^{\max}} f_{C_i}(c; t)f_{\Delta_i}(\delta; t) \,dc \,d\delta \nonumber \\
		&= \frac{1}{\sigma_{C_i} \sigma_{\Delta_i} Z_{C_i} Z_{\Delta_i}} \int_{\Delta_i^{\min}}^{\Delta_i^{\max}} \phi\left(\frac{\delta - \mu_{\Delta_i}}{\sigma_{\Delta_i}}\right) \nonumber \\
		& \qquad \times \left( \int_{\Theta \delta}^{C_i^{\max}} \phi\left(\frac{c - \mu_{C_i}}{\sigma_{C_i}}\right) \,dc \right) \,d\delta.
	\end{align}

	The inner integral evaluates to:
	\begin{equation}
		\sigma_{C_i} \left[ \Phi\left(\frac{C_i^{\max} - \mu_{C_i}}{\sigma_{C_i}}\right) - \Phi\left(\frac{\Theta \delta - \mu_{C_i}}{\sigma_{C_i}}\right) \right].
	\end{equation}

	Substituting and simplifying:
	\begin{align}
		R^{\text{H}}_i(t, \Theta) &= \frac{1}{\sigma_{\Delta_i} Z_{C_i} Z_{\Delta_i}} \Biggl[ \Phi\left(\frac{C_i^{\max} - \mu_{C_i}}{\sigma_{C_i}}\right) \sigma_{\Delta_i} Z_{\Delta_i} \nonumber \\
		&\quad - \int_{\Delta_i^{\min}}^{\Delta_i^{\max}} \phi\left(\frac{\delta - \mu_{\Delta_i}}{\sigma_{\Delta_i}}\right) \Phi\left(\frac{\Theta \delta - \mu_{C_i}}{\sigma_{C_i}}\right) \,d\delta \Biggr]
	\end{align}
	which yields Eq.~\eqref{eq:H_Theta_reliability_presented_earlier_lemma_formatted}.

	\section{Proof of Lemma~\ref{lemma:optimal_partitioning}}
	\label{app:proof_lemma3}

	Taking logarithms, we transform the product maximization into a sum maximization:
	\begin{equation}
		\max_{\boldsymbol{\alpha}} \sum_{i=1}^{N} \ln R_i(t, \alpha_i \Theta) \quad \text{subject to} \quad \sum_{i=1}^{N} \alpha_i = 1.
	\end{equation}

	Forming the Lagrangian:
	\begin{equation}
		\mathcal{L} = \sum_{i=1}^{N} \ln R_i(t, \alpha_i \Theta) - \lambda\left(\sum_{i=1}^{N} \alpha_i - 1\right).
	\end{equation}

	The first-order conditions $\partial \mathcal{L}/\partial \alpha_i = 0$ yield:
	\begin{equation}
		\frac{\Theta}{R_i(t, \alpha_i \Theta)} \cdot \frac{\partial R_i(t, \alpha_i \Theta)}{\partial (\alpha_i \Theta)} = \lambda \quad \text{for all } i.
	\end{equation}

	Since $\lambda$ is the same Lagrange multiplier for all devices, we have:
	\begin{equation}
		\frac{\partial R_i(t, \alpha_i \Theta)/\partial \Theta}{R_i(t, \alpha_i \Theta)} = \frac{\partial R_j(t, \alpha_j \Theta)/\partial \Theta}{R_j(t, \alpha_j \Theta)} \quad \forall\, i, j.
	\end{equation}

	This is the equal marginal log-reliability condition of Eq.~\eqref{eq:optimal_partitioning_condition}.

	\section{Proof of Lemma~\ref{lemma:series_feasibility}}
	\label{app:proof_lemma4}

	Since adding devices to a series system can only decrease reliability (each additional device introduces another potential failure point), the optimal selection strategy uses the most reliable devices available.

	Sorting the devices in descending order of reliability ensures that $\prod_{i=1}^{N} R_{(i)}(t, \Theta)$ is maximized over all possible size-$N$ subsets of the candidate pool. Any other subset would include at least one device with lower reliability, resulting in a smaller product.

	The constraint $R_{\text{series}}(t) = \prod_{i=1}^{N} R_{(i)}(t, \Theta) \geq \varepsilon$ directly yields the feasibility condition of Eq.~\eqref{eq:series_feasibility}. The maximum $N^*_{\text{series}}$ is the largest $N$ for which this product remains above the target $\varepsilon$, giving Eq.~\eqref{eq:series_max_general}.

	\section{Proof of Lemma~\ref{lemma:parallel_selection}}
	\label{app:proof_lemma5}

	Since adding devices to a parallel system can only increase reliability (each additional device provides another opportunity for success), the optimal selection strategy prioritizes the most reliable devices.

	For the $N$ most reliable devices (sorted in descending order), the system reliability is:
	\begin{equation}
		R_{\text{parallel}}(t) = 1 - \prod_{i=1}^{N}(1 - R_{(i)}(t, \Theta)) = 1 - P_N(t).
	\end{equation}

	The constraint $R_{\text{parallel}}(t) \geq \varepsilon$ is equivalent to:
	\begin{equation}
		1 - P_N(t) \geq \varepsilon \quad \Leftrightarrow \quad P_N(t) \leq 1 - \varepsilon.
	\end{equation}

	Since $P_N(t)$ decreases monotonically as $N$ increases (each additional term $(1-R_{(i)}) < 1$ reduces the product), the minimum $N^*_{\text{parallel}}$ is the smallest $N$ achieving this condition, yielding Eq.~\eqref{eq:parallel_exact}.

	\printbibliography

\end{document}